\documentclass[a4paper,onecolumn,11pt,accepted=2024-12-20]{quantumarticle}
\pdfoutput=1
\usepackage[style=numeric-comp,sorting=none]{biblatex}
\usepackage{graphicx}
\usepackage{bm}
\usepackage{amsthm}
\usepackage{amsmath}
\usepackage{amsfonts}
\usepackage[export]{adjustbox}
\newtheorem{thm}{Theorem}
\newtheorem{fct}{Fact}
\newtheorem{dfn}[thm]{Definition}
\newtheorem{lem}[thm]{Lemma}
\newtheorem{obs}[thm]{Observation}

\usepackage{physics}
\usepackage{appendix}\usepackage{hyperref}

\usepackage{color}

\newcommand{\je}[1]{\textcolor{black}{#1}}

\usepackage{phfcc}
\phfMakeCommentingCommand[initials={PhF},color=phfcc5]{phf}

\begin{document}

\title{Quantum complexity  phase  transitions 
in monitored random circuits}

\newcommand{\fu}{Dahlem Center for Complex Quantum Systems, Freie Universität Berlin, Berlin 14195, Germany}

\author{Ryotaro Suzuki}
\email{ryotaro.suzuki@fu-berlin.de}
\affiliation{\fu}

\author{Jonas Haferkamp} 
\affiliation{\fu}
\affiliation{School of Engineering and Applied Sciences,
Harvard University, Cambridge, MA 02318, USA}

\author{Jens Eisert}
\affiliation{\fu}

\author{Philippe Faist}
\affiliation{\fu}

\begin{abstract}
Recently, the dynamics of quantum systems that involve both unitary evolution and quantum measurements
have attracted attention 
due to the exotic phenomenon of 
measurement-induced phase transitions. The latter refers to a sudden change in a property of a state of $n$ qubits, such as its
entanglement entropy, depending on the rate at which individual qubits are measured.
At the same time, quantum complexity 
emerged as a key quantity for the identification of complex behaviour in quantum many-body dynamics.
In this work, we investigate the dynamics of the quantum state complexity in monitored random circuits, where $n$ qubits evolve according to a random unitary circuit
and are individually measured with a fixed probability at each time step.
We find that the 
evolution
of the exact quantum state complexity undergoes a phase transition
when changing the measurement rate.
Below a critical measurement rate, the complexity grows at least linearly in time until 
{saturating to a value $e^{\Omega(n)}$.}
Above, 
the complexity does not exceed 
$\operatorname{poly}(n)$.
%
In our proof, 
we make use of percolation theory to 
{find}
paths along which
an exponentially long quantum computation can be run below the critical rate, and to identify events where the state complexity is reset to zero above the critical rate.
We lower bound the exact state complexity in the former regime using recently developed
techniques 
{from} algebraic geometry.
Our results combine quantum complexity growth, phase transitions, and computation with measurements
to help understand the behavior of monitored random circuits and to make progress towards determining
the computational power of measurements in many-body systems.
\end{abstract}

\maketitle


\section{Introduction}


The evolution of quantum complexity in many-body quantum systems offers a new approach to understand
phenomena in quantum computation, quantum many-body systems, and black hole 
physics~\cite{aaronson2016complexityquantumstatestransformations}:
Complexity is able to capture the long-time behaviour of the quantum dynamics beyond the point
where many physical quantities, such as the entanglement entropy,
equilibrate to their limiting value~\cite{SusskindEntanglement,PhysRevLett.127.020501}.
Quantum complexity might be viewed as a 
measure of the time 
 for which a
suitably chaotic system has been evolving~\cite{SusskindEntanglement}:  Brown and Susskind conjectured
that complexity grows linearly in time for generic quantum dynamics of an $n$-qubit system until
saturating at times exponential in $n$~\cite{brown2018second}.  In contrast, the entanglement
entropy typically saturates after a time linear in $n$.
Versions of this conjecture have been proven in the context of
random circuits~\cite{brandao2021models,haferkamp2022linear,li2022short}.
Many recent results at the interface of quantum complexity and many-body systems have
mainly been driven by the central role that quantum complexity appears to play in the
\emph{anti-de-Sitter space/conformal field theory} (AdS/CFT) correspondence~\cite{PhysRev90.126007,PhysRevLett.116.191301,ChapmanMarrochioMyers,brown2018second,BigComplexity}:
The quantum complexity of the quantum state in a CFT is believed to
correspond to some physical property, such as the volume, of a wormhole contained in the
corresponding AdS space~\cite{SusskindEntanglement, susskind2016computational, Susskind2018PiTP_three,Belin2022PRL_anything}.
Overall, quantum complexity is a measure of the intricacy of the entanglement that is present
in an $n$-qubit state; its physical and operational interpretations in the context of
many-body physics are still being uncovered~\cite{brown2018second,PhysRevB.91.195143, PhysRevResearch.2.013323, PhysRevA.106.062417}.

To study the evolution of complexity of a CFT, 
one often resorts to the simpler 
model of
{local random quantum circuits}, in which the evolution of an $n$-qubit system is modeled by applying 2-qubit gates chosen at random on neighboring qubits.
Local random circuits are expected to reproduce a number of interesting features of chaotic
systems~\cite{fisher2023random, BlackHoles,HaydenPreskill, PhysRevX.8.041019, PhysRevB.102.064305}
while being technically more convenient to analyze than chaotic Hamiltonian 
dynamics~\cite{ActuallyChaoticHamiltonians}.
In the model of random quantum circuits, 
the quantum complexity 
has been proven to grow sublinearly in time until saturating at
times exponential in $n$~\cite{brandao2016local,roberts2017chaos,brandao2021models},
using the toolbox of unitary $t$-designs~\cite{PhysRevA.80.012304,2007JMP....48e2104G}, where the quantum complexity is lower-bounded by $\Omega(\tau^a)$ with time $\tau$ and a positive number $a<1$.
%
%
The linear growth of the exact circuit complexity 
for local random quantum circuits 
was eventually proved 
in Ref.~\cite{haferkamp2022linear} by exploiting
geometric arguments, where the exact complexity is lower-bounded by $\Omega(\tau)$ with time $\tau$.
More precisely,
{the toolbox of algebraic geometry 
enables a quantification
of the dimension of the set of all possible unitaries 
that can be achieved with a fixed
number of gates in a specific circuit layout.  This dimension, called
\emph{accessible dimension}, yields a lower bound the exact quantum complexity of the
random circuit.  The main result of Ref.~\cite{haferkamp2022linear} is a consequence of
the fact that the accessible dimension grows linearly in time until
saturating at a time exponential in $n$ (cf.\@ also simplified proofs in
Ref.~\cite{li2022short}).  We heavily rely upon}
this powerful mathematical toolkit 
in this work.

A different line of research at the interface of computational complexity theory
and many-body physics concerns 
{complexity phase transitions}.
The latter refer to situations where the complexity of solving a particular task undergoes
a sudden and drastic change when a parameter in the problem is 
varied.
Such complexity phase transitions 
\je{have initially been} discussed 
when studying the hardness regime of the random $k$-SAT
problem~\cite{DBLP:conf/ijcai/CheesemanKT91,DBLP:conf/aaai/MitchellSL92}.
More recently, {the complexity of 
classically simulating quantum circuits 
\je{has been} found
to undergo a 
transition}
for \emph{instantaneous 
quantum polynomial} (IQP)
circuits~\cite{bremner2011classical,fujii2016computational,park2022complexity},
linear optical circuits~\cite{aaronson2011computational,deshpande2018dynamical,
PhysRevLett.129.150604,PhysRevResearch.4.L032021},
random quantum circuits~\cite{bouland2019complexity,PhysRevX.12.021021},
and dual-unitary  circuits~\cite{PhysRevB.100.064309,PhysRevLett.123.210601,PhysRevB.101.094304,%
PhysRevLett.126.100603,Suzuki2022computationalpower}. 
Such 
transitions are of particular interest for
{drawing 
\je{and delineating the}
boundaries between the power of classical and quantum computing \cite{PhysRevA.62.062311}.}

\begin{figure*}[tbh]
\centering
\includegraphics[width=158mm]{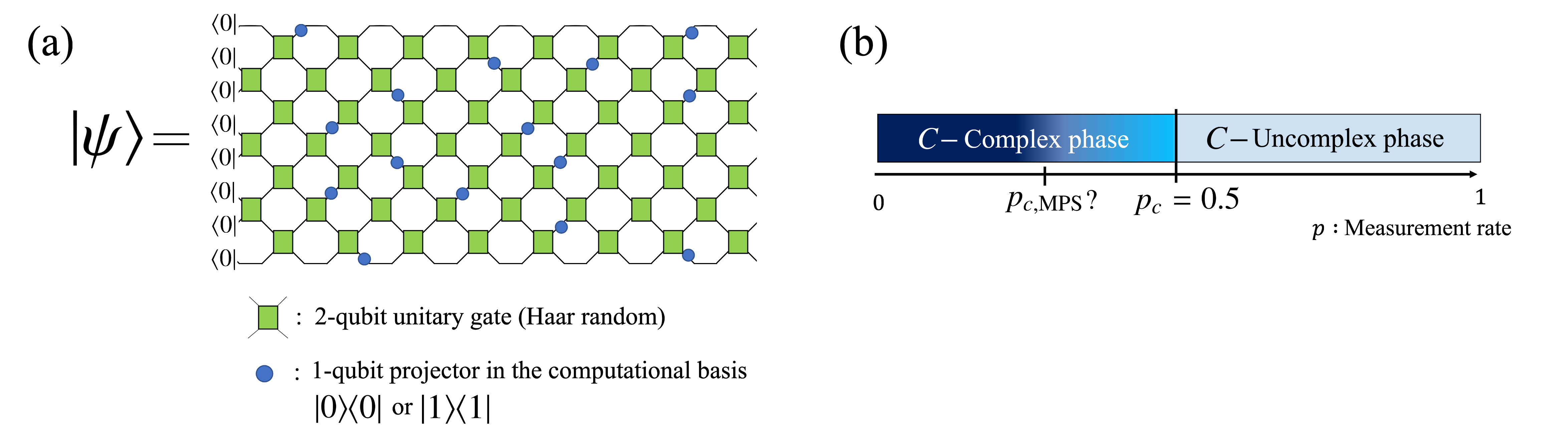}
    \caption{Our setup and the summary of the result.
    (a)~We consider monitored random circuits consisting of two-qubit Haar-random gates (green boxes)
    arranged in staggered layers, where at each time step the individual qubits undergo a measurement
    in the computational 
    basis (thick blue points) with probability $p$.
    The state vector $\ket\psi$ is obtained by applying the circuit onto the computational basis
    state vector $\ket{0^n}$ and conditioning on all the measurement outcomes.
    (b)
    We find that 
    the complexity phase diagram of the monitored random circuit
    exhibits a phase transition at the critical measurement
    rate $p_c=0.5$. The measure of complexity,
    $C(\ket\psi)$, is defined as the
    minimal number of two-qubit gates required to prepare $\ket\psi$ exactly, in any circuit layout.
    In the $C$-complex phase ($p<p_c$), the complexity $C$ 
    grows {at least} linearly 
    {until saturating to a value that is exponential in the system size}.
    In 
    the $C$-uncomplex phase ($p>p_c$), the quantity $C$ {saturates to a value $\operatorname{poly}(n)$ after a time no more than $O(\log (n))$}.
    {The result regarding the uncomplex phase agrees with earlier numerical results on the area law 
    \cite{AreaReview}
    of R\'enyi-0 entropy in the regime $p>p_c$ in Ref.~\cite{PhysRevX.9.031009}.}
    Moreover, another earlier numerical  result~\cite{
 PhysRevB.101.104301} points 
 to a description of the state vector $\ket\psi$
    in terms of a \emph{matrix-product state} (MPS)~\cite{mpsGarcia2007} with $\operatorname{poly}(n)$ bond dimension
    in the region $p \geq p_{c,\mathrm{MPS}}$,
    with $p_{c,\mathrm{MPS}} < 0.5$.  Given that a robust measure 
    of complexity would
    saturate after a time $\sim \operatorname{poly}(n)$ in this region, it is likely the region
    $p_{c,\mathrm{MPS}} < p < p_c$  yields examples
    of states generated by monitored random circuits whose exact complexity grows to large values yet remain
    close in trace norm to a state of low complexity.}
    \label{fig:intro}
\end{figure*}

Moreover,
the effect of {measurements} on the dynamics of a complex many-body quantum system
has drawn significant interest in the many-body
physics community.  A common model combining measurements and unitary evolution is a \emph{monitored random quantum circuit} on $n$ qubits and with measurement rate $p\in[0,1]$:
At each time step, randomly chosen two-qubit gates are applied between neighboring qubits;
furthermore, each individual qubit undergoes a measurement in the computational basis with a
probability $p$.  This simple model has recently attracted substantial attention from the condensed matter
physics community because such circuits may exhibit \emph{measurement-induced phase transitions}%
~\cite{li2018quantum,  PhysRevX.9.031009, chan2019unitary, li2019measurement, PhysRevB.101.104301, PhysRevX.10.041020, PhysRevB.101.104302, PhysRevLett.125.030505,  PhysRevX.11.011030, lavasani2021measurement,  PRXQuantum.2.010352, PhysRevResearch.3.023200,  PhysRevResearch.4.043212,
PhysRevX.12.041002,
10.21468/SciPostPhys.7.2.024, PhysRevResearch.2.013022, PhysRevB.102.054302, PhysRevLett.127.140601, PhysRevLett.127.140601, PhysRevLett.128.010603, kelly2210coherence, niroula2023phase}. 
The latter are an exotic type of phase transition that depends on the rate $p$ at which
measurements are 
performed:  The state's entanglement entropy then commonly
transitions from a scaling in the area of the region considered \cite{AreaReview} at high $p$
(the {area law phase}) to a scaling in the volume of the region at low $p$
(the {volume law phase}).%


The goal of our work is to combine the ideas of (i)~complexity growth in many-body systems,
(ii)~complexity phase transitions, and (iii)~measurement-induced phase transitions,
to prove the existence of a sharp transition in the evolution of quantum complexity
in monitored quantum circuits depending on rate at which measurements are applied.
We thereby introduce the distinct notion of quantum state complexity into the study of monitored quantum circuits.
%

Specifically, we prove rigorously that the growth of the exact state complexity
in a monitored random circuit on $n$ qubits
makes a sharp transition at a critical rate $p_c = 0.5$ at which measurements
are applied (sketched in Fig.~\ref{fig:intro}).
Below the threshold, the quantum complexity grows 
 {at least} linearly in time until saturating
{to a value $e^{\Omega(n)}$}
 (the \emph{complex phase}).
Above the threshold, the state's complexity 
{saturates to a value $\operatorname{poly}(n)$ after a time no more than $O(\log (n))$}
(
the \emph{uncomplex phase}).
We quantify the state's quantum complexity in terms of the number of two-qubit unitary
gates required to prepare that state exactly. 


{We establish a framework of the study of complexity of monitored random circuits as follows.
We draw deep inspiration from the seminal work on measurement-induced phase transitions
in the dynamics of entanglement~\cite{PhysRevX.9.031009}, including the use of techniques from
percolation theory~\cite{grimmett1999percolation},
while adapting to 
the techniques to lower bound the exact quantum circuit complexity using 
semi-algebraic geometry of  Ref.~\cite{haferkamp2022linear}.}
%
The complexity phase transition that we find 
concerns the quantum complexity of the output state of the monitored circuit,
and might be of different nature than the phase transition
in the classical complexity of sampling 
outcomes from random
circuits~\cite{PhysRevX.12.021021} and monitored linear optical circuits~\cite{PhysRevResearch.4.L032021}.
Our results reinforce monitored random circuits as a promising 
model to investigate quantum complexity phase transitions
{and the influence of measurements on the complexity of a quantum circuit's output state.}

The remainder of this work is organized as follows.
In Section~\ref{sec:review}, we 
review 
monitored quantum circuits and methods of lower-bounding the state complexity.
In Section~\ref{sec:main_result}, we summarize our main results, 
discuss their core implications, and sketch our proof strategy.
In Section~\ref{sec:exact complexity}, we give a proof of our main result.
Section~\ref{sec:conclusion} is devoted to conclusion and discussion.

\section{Setting}
\label{sec:review}


In this section, we {review the definitions of}
monitored random quantum circuits, of the exact state complexity, 
and of the accessible dimension.

\subsection{Monitored random quantum circuits} \label{sec:MRQC_def}

Throughout this work, we consider
{a system of $n$ qubits.  The qubits might be realized, for instance, as individual spins of a quantum many-body system.} 
For {technical convenience}, 
we assume that $n$ is an even number.
The computational basis 
{of the system} is denoted by $\ket{i_1, i_2,\dots, i_{n}}$,
where $i_j=0,1$ indicates {the state of the $j$-th qubit}.
A \emph{monitored random quantum circuit with measurement rate $p\in[0,1]$}
is a quantum circuit with staggered layers of two-qubit
gates on nearest neighbors, or the brick-wall architecture, in which each qubit has a probability $p$ at each time step
to be measured in its computational basis and be projected into the resulting outcome
[Fig.~\ref{fig:intro}(a)].
It is defined as 
\begin{equation}  
V^M(t):=\prod_{\tau=1}^{t/2} M(2\tau)U^{(e)}(2\tau)M(2\tau-1)U^{(o)}(2\tau-1), \label{eq.MRQC}
\end{equation}
where
\begin{align}
U^{(o)}(2\tau-1) & :=\prod_{i=1}^{\frac{n}{2}}U_{2i-1,2i}(2\tau-1), \label{eq.1Dqc/o} \\
U^{(e)}(2\tau) & :=\prod_{i=1}^{\frac{n}{2}-1}U_{2i,2i+1}(2\tau),\label{eq.1Dqc/e}\\
M(\tau) & :=\prod_{i=1}^{n}M_{i}(\tau).
\end{align}
Here, $t$ is an even number, $U_{i,j}(\tau)$ is a {Haar-random unitary gate}
acting on qubits $i$ and $j$ at time $\tau$, and {$M_{i}(\tau) \in
\{ \sqrt{1-p}I_i, \sqrt{p} \ketbra{0}{0}_i, \sqrt{p} \ketbra{1}{1}_i \}$. 
The latter are Kraus operators of the channel that implements a measurement in the computational
basis with probability $p$.}
We say that the qubit $i$ is \emph{measured} at time $\tau$ if $M_{i}(\tau)$
is either $\sqrt{p} \ketbra{0}{0}_i$ or $\sqrt{p} \ketbra{1}{1}_i$. {The \emph{measurement configuration}}
$M:= \{  M_i  (\tau) \}_{i, \tau}$, is the collection of all
{measurement outcomes at each space-time point of the circuit}.
{By construction, $M$ contains all the information about the layout of the
circuit, including $n$ and $t$, along with which qubits were measured at which time,
and what the projective measurement outcomes were.}
%
Note that the time evolution operator $V^M(t)$ in 
Eq.~(\ref{eq.MRQC}) is \emph{not} unitary, i.e., \begin{equation}
V^M(t)^{\dagger} \neq V^M(t)^{-1}, 
\end{equation}
except in the situation when the measurement rate $p$ is exactly zero.
{That $V^{M}(t)$ is not unitary corresponds to the fact that we measure the system and condition the
evolution on the measurement outcomes specified by $M$.}

Our results concern the output of a monitored quantum circuit when it is applied
onto the initial state vector  $\ket{0^n}$.
The state vector $V^M(t) \ket{0^n}$ represents the unnormalized
output of the monitored quantum circuit, 
projected according to the measurement configuration $M$.
Its squared norm $\expval{V^M(t)^{\dagger}  V^M(t)}{0^n}$ is
the probability that a measurement configuration $M$ is observed
for fixed choices of gates $U_{i,j}(\tau)$. 
Our results concern the complexity of the normalized
output quantum state vector
\begin{align}
    \ket{\phi^M} := \frac{V^M(t) \ket{0^n}}{\lVert V^M(t) \ket{0^n} \rVert} \ .
  \label{eq:output-state-random-monitored-q-circuit-normalized}
\end{align}
This state is the output of the monitored quantum circuit after 
conditioning on the
measurement outcomes $M$.


\subsection{State complexity} \label{subsec:Complexity}

The complexity of a quantum state vector $\ket\psi$ refers to the minimal number of elementary operations,
such as two-qubit gates, 
that need to be composed in order to prepare $\ket\psi$ starting from the reference
state vector $\ket{0^n}$. 
The complexity of a state is ordinarily defined by considering two-qubit unitary gates as the elementary
operations.  We call this complexity measure the \emph{$C$-complexity}:
\begin{dfn}[Exact $C$ state complexity] \label{def:exact_state_complexity_C0}
  The \emph{$C$ state complexity} of a normalized state vector $\ket{\psi}$ is the minimal number of
  two-qubit gates required to prepare $\ket\psi$ from the state vector $\ket{0^n}$.  The gates can be any elements
  of $SU(4)$ and the circuit may have any chosen connectivity.
\end{dfn}
%
We 
also consider a 
stronger notion of complexity in which the elementary
operations also include measurements with {post-selection}~\cite{aaronson2005quantum}.
A post-selected circuit is defined as a quantum circuit consisting of two-qubit unitary gates and single-qubit measurements in the computational basis where the measurement outcomes are post-selected to the desired measurement outcomes, for example, $0$ for all outcomes.
{At any time in a post-selected circuit}, arbitrary qubits, for instance the $i$-th qubit,
of a state vector $\ket{\psi}$ can be measured in the computational basis and be post-selected
to the desired measurement outcome $0$, resulting in the state
$({\lVert \bra{0}_i \ket{\psi} \rVert})^{-1}{\ketbra{0}_i\ket{\psi}}$.
The exact state complexity $C$ with post-selected circuits is defined as follows:

\begin{dfn}[Exact $C_m$ state complexity] \label{def:exact_state_complexity}
  For a state vector $\ket{\psi}$
  with $0 < \braket{\psi}{\psi} \leq 1$,
  the exact $C_m$ state complexity $C_m (\ket{\psi})$ is the minimal number of two-qubit gates in an arbitrary post-selected circuit that prepares $({\lVert \ket{\psi} \rVert} 
  )^{-1}\ket{\psi}$ from the initial state vector $\ket{0^n}$.  The post-selected circuit consists of 
  two-qubit unitary gates with arbitrary connectivity and where an arbitrary number of single-qubit
  computational basis measurements can be applied, with post-selection on a desired outcome, at
  any space-time points of the circuit.
  %
\end{dfn}

The set of post-selected quantum circuits includes unitary circuits as a special case,
implying that
the measure of complexity $C_m$ 
is a lower bound on the usual state 
complexity $C$. 
%


\subsection{Accessible dimension} \label{subsec:AD}

{
%

The \emph{accessible dimension} \cite{haferkamp2022linear} has been defined as the dimension of
the set of all possible unitary circuits that can be achieved with a fixed circuit layout, by
varying the individual choices of the gates in that circuit.
Here, we adapt this definition to our setting, and show that it serves lower-bounds, analogously to the proof in Ref.~\cite{haferkamp2022linear}, on the $C$- and $C_m$- complexity of a monitored random circuit below the critical measurement probability.
For a monitored random circuit with a fixed measurement configuration $M$, 
we define the 
\emph{contraction map} 
from a collection of two-qubit unitary gates to the output state as 
\begin{align}
    & F^M: [ \textrm{SU}(4)]^{ \times R} \xrightarrow{} B_1^{2 \times 2^n} \subset \mathbb{C}^{2^n},\\
    & F^M (U_1, U_2, \dots, U_R)=V^M(t) \ket{0^n},
\end{align}
where $B_1^{2 \times 2^n}$ is the real unit ball with the center at the origin, where
$R$ is the total number of two-qubit unitary gates
in the monitored random circuit specified through $M$,
and where each two-qubit unitary gate in Eq.~(\ref{eq.MRQC})
is set to the corresponding unitary $U_j$.  
%
{That the image of $F^M$ includes sub-normalized $n$-qubit states is a consequence
of $V^M(t)$ not being unitary.}
We denote the image of $F^M$ by $\mathcal{S}^M$, that is, the set of all output
states generated by the monitored  random quantum circuit with $M$.
{(See additional technical details in Appendix~\ref{sec:App:Accesdim}.)}


{We define the \emph{rank} of $F^M$ as the number of independent degrees of freedom required to
specify a perturbation of the image of $F^M$ when we perturb the gates $\{ U_1, \ldots U_R \}$.
More specifically,}
the rank of $F^M$ at a point $\{ U_1, U_2, \dots, U_R \}$,
 {denoted by $\rank_{U_1,\ldots,U_R}(F^M)$}, is defined by the dimension of the real linear space spanned by the set of output state vectors 
\begin{align}
\{ F^M (U_1, \dots,  (\alpha \otimes \beta) U_j, \dots , U_R) \}_{j, \alpha, \beta}, \label{eq:vectors}
\end{align}
where $j \in \{1, 2, \dots , R \}$ and $\alpha$, $\beta$ $\in \{ I, X, Y, Z  \}$ are Pauli operators
such that $(\alpha, \beta) \neq (I, I)$.
{We then define the}
\emph{accessible dimension} as 
the maximal rank of $F^M$ over all unitary gates:

\begin{dfn}[Accessible dimension] \label{dfn:AD}
  For a monitored random quantum circuit with a measurement configuration $M$, 
  the accessible dimension {$d_M$} is the maximal rank 
  of $F^M$ over  all two-qubit unitary gates $\{ U_1, U_2, \dots, U_R \}$, where $U_j \in \textrm{SU}(4)$.
\end{dfn}
%
{A strategy to lower bound the accessible dimension $d_M$ is to lower bound the rank of $F^M$ at
any chosen point $\{ U_1, \ldots U_R \}$.}
{The accessible dimension {$d_M$} is also the dimension of the set
  $\mathcal{S}^M$ (see Appendix~\ref{sec:App:Accesdim}).}
We prove that the complexity measure $C_m$ is lower bounded in terms of $d_M$, which is analogous to the proof in Ref.~\cite{haferkamp2022linear}.
\begin{lem}[Complexity by dimension] \label{lem:complexity_by_dimension}
  Let $\ket{\psi} \in \mathcal{S}^M$
  be distributed according to the output of a monitored random quantum circuit
  with a fixed measurement configuration $M$,
  in which all unitary gates are chosen at random
  from the Haar measure.  Then
  $C_m(\ket{\psi}) \geq 
  (  {d_M} -3n-2)/11$ with unit probability.
\end{lem}

In the above lemma, given a measurement configuration of a monitored circuit $M$, the accessible dimension has been shown to take the maximum value over all unitary gates except for a measure zero set. The above lemma serves in our proof to reduce the problem of finding
a lower bound on $C_m$ for a monitored random quantum circuit to finding a lower bound on $d_M$.

\section{Main result: Complexity phase transition in  monitored  random quantum circuits}

\label{sec:main_result}


We prove that both of the $C$- and  $C_m$-complexity of the output state of 
a monitored random quantum circuit exhibit
a phase transition at a critical measurement probability $p_c=0.5$.

\begin{thm}[Complexity growth in monitored circuits]\label{thm:phase_transition_Cs}
Let $\ket{\psi}$ be the output state vector of the monitored random circuit with measurement rate $p$, {conditioned on the outcomes of the measurement that were applied in the monitored circuit}.
    If $p < p_c$, $C(\ket{\psi})$ and $C_m(\ket{\psi})$ grow {at least} linearly and  linearly in $t$, respectively, until 
    {they saturate to values $e^{\Omega(n)}$},
    with probability $1-e^{-\Omega(n)}$. 
    If $p>p_c$, and for any $0<\epsilon<1$,
we have $C(\ket\psi) \leq \operatorname{poly}(n/\epsilon)$ and $C_m(\ket\psi) \leq O\bigl(n\log(n/\epsilon)\bigr)$ except
with probability at most $\epsilon$.
\end{thm}



{Our bounds on 
both complexities do not depend on the specific measurement outcomes $M$,
even though the output state vector $\ket{\psi}$ is conditioned on $M$.}

Our proof exploits techniques from {percolation theory}~\cite{grimmett1999percolation}
to prove a sharp transition between these two regimes at the critical measurement rate $p_c=0.5$.
Above this rate, measurements percolate across the width of the circuit,
periodically resetting the state's complexity.
This implies the upper bound of the complexity by $\operatorname{poly}(n)$.
Below the critical rate, it turns out multiple paths without any measurements
can percolate along the length of the circuit,
supporting a computation whose complexity grows linearly in time until
times exponential in $n$.
The growth of $C(\ket\psi)$ in the regime $p < p_c$ follows from the general bound
$C_m(\ket\psi) \leq C(\ket\psi)$.


Our core technical result is 
a lower bound on the accessible dimension
of a monitored random quantum circuit in the regime 
$p<p_c$ 
By Lemma~\ref{lem:complexity_by_dimension}, this bound immediately
translates into a corresponding bound on the $C_m$-complexity.

\begin{lem}[Growth of the accessible dimension in monitored circuits]\label{thm:phase_transition_AD}
    If $p < p_c$, $d_M$ grows linearly in $t$ until an exponential time $t=e^{\Omega(n)}$ with probability $1-e^{-\Omega(n)}$.
\end{lem}

{We now provide a sketch of the proof of Lemma~\ref{thm:phase_transition_AD}.  Two
separate arguments are developed in the regimes $p>p_c$ and $p<p_c$.
In the regime $p>p_c=1/2$, percolation theory states that measurements will regularly
percolate throughout the width of the circuit, resetting the state vector to $\ket{0^n}$ along
those paths (Fig.~\ref{fig:crossings}(a)).  Such measurement percolation occurs within the last $n$ layers of gates in
the monitored circuit with probability $1-e^{-\Omega(n)}$, meaning that
the set of output states
of the circuit cannot have 
the $C_{m}$-complexity cannot exceed $O(n^2)$.}
{The argument can be further reinforced to upper-bound $C$-complexity by
$\operatorname{poly}(n)$}, 
and to bound the $C$ and $C_m$ complexity measures
in the case where the tolerated
failure probability is arbitrary.

In the regime $p<p_c=1/2$, we lower-bound the accessible dimension as follows.
We first show that for a fixed configuration of measurements $M$,
there are paths without any measurements that percolate throughout the length of the
circuit.
We call such paths \emph{measurement-free paths}. 
Then we show that these paths can be used to run an exponentially long quantum
computation.
The main challenge is to construct an embedding of an arbitrary quantum
circuit on $\Omega(n)$ qubits and of depth $\Omega(t)$
into the monitored random quantum circuit with a fixed
configuration of measurements $M$.
The main idea of the embedding is to associate
one qubit of the $\Omega(n)$ sized circuit to a measurement-free path, and to
choose the gates of the monitored random circuit so that they
ensure the qubit's information is carried along the measurement-free path and
that they implement the gates of the $\Omega(n)$ sized circuit on those qubits.
There are two specific challenges 
that one faces when {constructing this embedding}:
One must show that (a)~the computation can proceed even if a measurement-free path
is not \emph{causal}, i.e., if it momentarily wraps back in time by following legs
between gates in a direction opposite to the circuit's time direction, and
(b)~two-qubit gates can be implemented between two such paths (Fig.~\ref{fig:crossings}(b)).
Challenge~(a) is addressed as follows. If a measurement-free path follows a leg of
a unitary gate in a direction contrary to the circuit's forward time direction,
then we can exploit the existence of a
measurement immediately after that gate to teleport the qubit being carried
by the path further along the path, even if the information is carried
backwards with respect to the circuit's direction.  This is possible because the
measurement configuration $M$ is fixed, meaning that the measurement immediately after
the gate has a predetermined outcome onto which the state is projected.
To address challenge~(b), we exploit the fact that paths with no measurements
also percolate vertically across the width of the circuit; these paths can be used to
implement a CNOT gate across two measurement-free paths using a teleportation-based scheme.

Both arguments addressing challenges~(a) and~(b) rely on the existence of measurements on
certain qubits that are neighboring the measurement-free path.  Yet such measurements
might not always exist at the desired locations.
We prove that for any measurement configuration $M$, one can always select additional
qubits to be measured without increasing the accessible dimension
of the monitored random circuit.
Therefore, should a measurement at a given location be required by our
embedding scheme, it can always be added if necessary
while still yielding a lower bound on the accessible dimension of the monitored
circuit in the original measurement configuration.
We believe that the following lemma might be of independent interest, as it provides
a rigorous quantitative statement about the impossibility of measurements to increase
a quantity, the accessible dimension, which is a proxy quantity for complexity
{for monitored random circuits}.
\begin{lem}[Measurements cannot increase the accessible dimension]
\label{lem:measurement cannot increase dimension}
{Let $M$ be a measurement configuration, and let 
$M'$ be a configuration obtained by changing some space-time locations
in $M$ from being unmeasured to being measured.
Then $d_{M'} \leq d_{M}$.}
\end{lem}
{Intuitively, the dimension of the set of states generated by a random monitored
circuit for a given measurement configuration cannot increase if one inserts an additional
projector in the circuit.
We present a proof of this statement
as Lemma~\ref{lem:measurement cannot increase dimension2} in
Appendix~\ref{app: measurement and dimension}.}

\begin{figure*}[tb]
\begin{center}
\includegraphics[width=128mm]{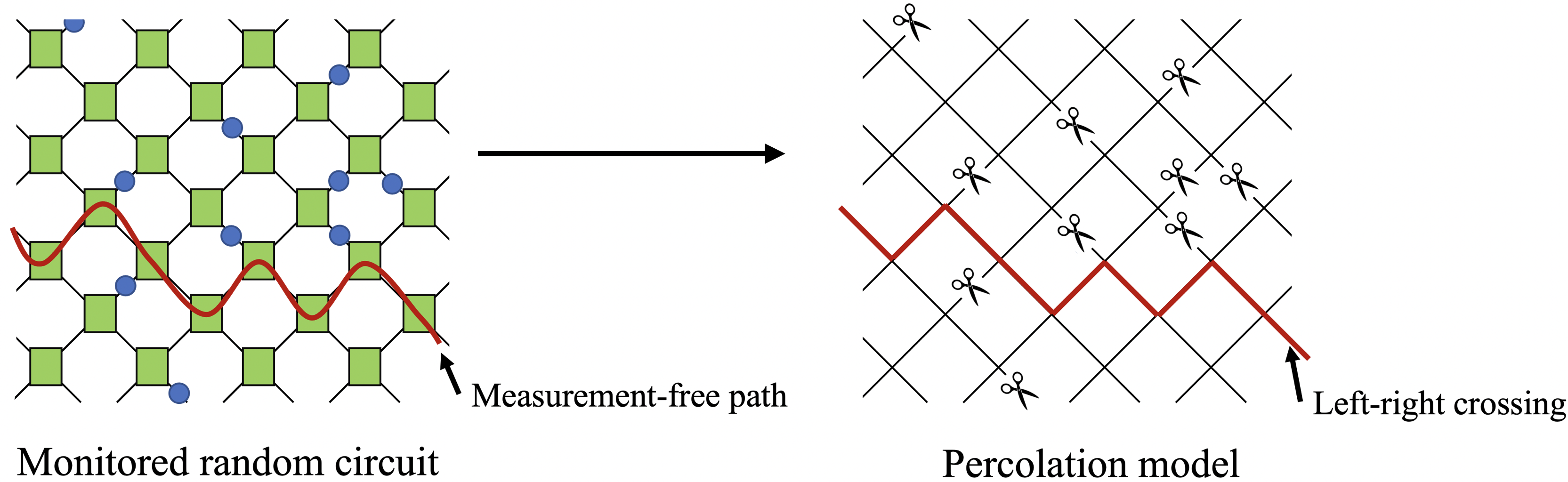}
\end{center}
    \caption{A mapping from the bond structure of a monitored random circuit to a percolation model. Each measured and unmeasured bond is mapped to closed (with scissors) and open edge (without scissors), respectively. }
    \label{fig:Circuit_to_percolation}
\end{figure*}

\section{Proof of the main result} \label{sec:exact complexity}

In this section, we prove 
Theorem~\ref{thm:phase_transition_Cs}.
{A central ingredient of our proof is the use of techniques from percolation theory. We briefly
review these techniques in Section~\ref{sec:percolation-theory}.
We then apply these techniques in Section~\ref{sec:proof-Cs-uncomplex-phase}
to obtain an upper bound on the $C$ complexity in the regime $p>p_c$.
Finally, we complete the proof of Theorem~\ref{thm:phase_transition_Cs} in
Section~\ref{sec:proof-Cs-complex-phase}
by proving a lower bound on the $C$ complexity in the regime $p<p_c$.}

\subsection{Percolation theory}
\label{sec:percolation-theory}

{In percolation theory, we consider a graph whose edges can be in one of two states, open or closed,
where the state of each edge is chosen to be open or closed
independently with probability $q$ and $1-q$, respectively~\cite{grimmett1999percolation}.}
{Bond percolation theory} is 
concerned with the existence or absence of a path
consisting of connected open edges in the 
graph. 
A well-studied setting is the existence of a 
path that crosses from left to right in a $L\times L$ square lattice
{while passing only through open edges}.
In the large $L$ limit, there is a critical probability $q_c$ below which 
there does not exist a left-right crossing with the probability $1-e^{-\Omega(L)}$, but above which 
such crossings appear with probability $1-e^{-\Omega(L)}$.
Moreover, for a square lattice in two spatial dimensions, 
the critical probability is 
$q_c=1/2$. 
We refer to Appendix~\ref{sec:App:percolation} for a more in-depth review of percolation theory,
including percolation on rectangular lattices.


Our application of percolation theory follows similar techniques used to compute 
the R\'enyi-0 entropy in Ref.~\cite{PhysRevX.9.031009}.
In order to formally apply techniques from percolation theory to monitored
random quantum circuits, we
map a monitored random circuit to a 
graph with edges that are randomly open or closed.
%
%
%
We define a 
{graph} by mapping each two-qubit unitary gate and its unmeasured bonds to a vertex and the open edges incident with it, respectively (Fig.~\ref{fig:Circuit_to_percolation}). 
With this mapping, the measurement rate $p$ is equal to the probability of closing an edge $1-q$.
{Moreover, percolation results for the square lattice extend naturally to the diagonally tilted square lattice
as in Fig.~\ref{fig:Circuit_to_percolation}, given that percolations from the left to the right of the
tilted lattice can be constructed from left-right and top-bottom percolations on the original, untilted
lattice (cf.\@ Appendix~\ref{sec:App:percolation}).}


\begin{figure*}[htb]
\begin{center}~
\includegraphics[width=150mm]{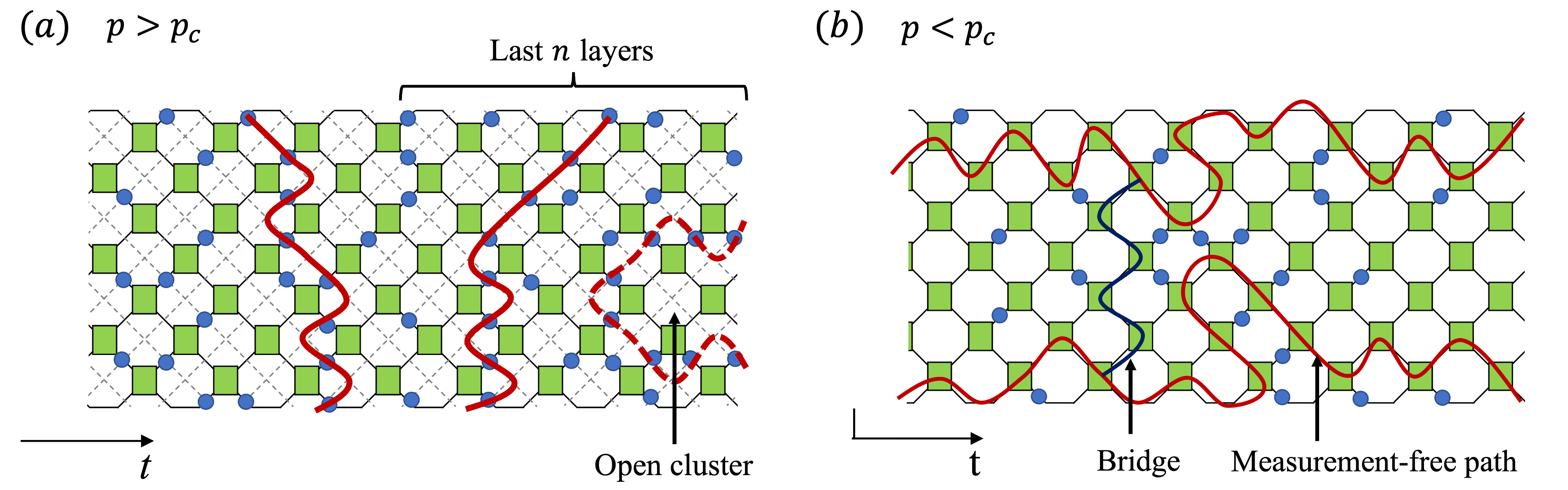}
\end{center}
    \caption{Monitored random circuits 
    above and below the critical measurement probability. (a) 
    Above the threshold, paths
    of measurements cut across the circuit from top to bottom. Their effect is to reset
    the state along the path to a product state, whose complexity vanishes.
    The region delimited by the red broken line is an example of an \emph{open cluster}. 
    (b)  Below the threshold, a linear number of measurement-free paths cross from the beginning
    to the end of the circuit.  These paths 
    can be used to embed a unitary circuit into the monitored circuit.
    A bridge (black line) 
    is  a vertical path of open edges
    used to implement two-qubit gates 
    in the
    embedded unitary circuit.}
    \label{fig:crossings}
\end{figure*}

\subsection{The uncomplex phase} 
\label{sec:proof-Cs-uncomplex-phase}

As a warm-up and to build additional intuition with the proof techniques we use, 
we first provide 
a simple upper bound on the $C_m$-complexity in the regime $p>p_c$:
%
Consider a circuit of depth $t>n$, and consider the last $n$ layers of that circuit.
Our strategy is to use percolation theory to conclude that there exist measurements that cut through the
width of the circuit in those last $n$ layers, resetting the state vector to $\ket{0^n}$ at the location of those
measurements [Fig.~\ref{fig:crossings} (a)].
We apply percolation theory to the dual lattice of the percolation model introduced
in Fig.~\ref{fig:Circuit_to_percolation}, depicted in Fig.~\ref{fig:crossings} (a).
For $p>p_c$, percolation theory states that with probability $1-e^{-\Omega(n)}$ there exist
paths of measurements in the dual lattice that connect the top of the circuit with the bottom side of the circuit.  For such a path, there is no unmeasured bond connecting the gates on the left side of the path to the gates on the right side of the path. (This
property would not have been guaranteed had we applied percolation theory directly
to the graph in Fig.~\ref{fig:Circuit_to_percolation} rather than to its dual lattice.)
The measurements therefore reset the state vector along the path to $\ket{0^n}$.
Since there are at most $O(n^2)$ gates after this path, both the accessible dimension $d_M$ as well as the output state complexity $C_m(\ket\psi)$ cannot exceed $O(n^2)$.
Therefore, if $p>p_c$, then $C_m \leq O(n^2)$ except with probability $e^{-\Omega(n)}$.

We now present our the part of the proof of Theorem~\ref{thm:phase_transition_Cs} pertaining
to the uncomplex phase. 
%
{
Our proof proceeds by 
upper bounding the size of regions consisting of connected open edges,
or \emph{open clusters}, on the graph in Fig.~\ref{fig:Circuit_to_percolation}. 
Open clusters 
correspond to 
connected bonds of gates in the circuit which are not measured (inside the broken line in Fig.~\ref{fig:crossings} a). 
The output state only depends on the unitary gates whose bonds are in the open
clusters and contain the boundary at the final time. 
Indeed, single-qubit  measurements at the boundary
of the open clusters reset each qubit to $\ket{0}$. 
There are 
no more than 
$n/2$ open clusters containing the bonds at the final time.
We now 
upper bound the size of such open clusters by $O(\log (n))$. 
Let $\{C_i\}_1^m$, where $m \leq n/2$, be the set of the distinct  open clusters containing the bonds at the final time and $|C_i|$ be the number of edges, or bonds, in $C_i, i \in \{1, 2,  \je{\dots}, m\}$. Then, the following lemma holds.
\begin{lem}[Small unmeasured regions]
  \label{lem:small_unmeasured_regions}
  Assume $p>p_c$. For any $0< \epsilon <1$, it holds:
  \begin{align}
      |C_i|=O \left( \log \left( \frac{n}{\epsilon}\right) \right),
  \end{align}
  for all $i \in \{1, 2,  \je{\dots}, m\}$, with probability $1-\epsilon$.
\end{lem}
We give a proof of the above lemma in
Appendix~\ref{App:subcritical} (stated \je{there} 
as Lemma~\ref{lem:Small_open_clusters}).  
Because of Lemma~\ref{lem:small_unmeasured_regions}, the output state is generated by a $O(\log(n/\epsilon))$-depth post-selected quantum circuit, implying that
$C_m 
\leq O(n\log(n/\epsilon))$ with probability $1-\epsilon$.
Moreover, it indicates that the Schmidt rank of the output state in any bi-partition is $\operatorname{poly} (n/\epsilon)$, implying that the output state can be efficiently represented by a 
\emph{matrix product state} 
(MPS)~\cite{mpsGarcia2007}, and therefore it is prepared by a unitary circuits with $\operatorname{poly}(n)$ complexity \cite{PhysRevLett.95.110503, cramer2010efficient}.
Overall, for the output state vector $\ket{\psi}$, our argument gives upper bounds
$C(\ket{\psi}) 
\leq \operatorname{poly} (n/\epsilon)$.}

This proof also recovers the upper bound for $C_m$ obtained in our initial percolation
argument (cf.\@ warm-up proof above) when $\epsilon$ is chosen exponentially small.
Plugging $\epsilon = e^{-cn}$ for fixed $c>0$ yields the upper bound
$C_m \leq O(n\log(n) + cn^2) = O(n^2)$.


\subsection{The complex phase} 
\label{sec:proof-Cs-complex-phase}

{The $C$-complex phase refers to the phase in which the $C$ complexity
  grows {at least} linearly until saturating to a value $e^{\Omega(n)}$.  We show that this phase
  occurs in monitored random circuits whenever $p<p_c$.}



Our proof proceeds as follows.  For a fixed measurement configuration $M$,
the goal is to prove a lower bound on the accessible dimension $d_M$ in order
to apply Lemma~\ref{lem:complexity_by_dimension}.
%
%
The strategy to lower bound $d_M$ is to show that, for some $m=\Omega(t)$, it is
possible to embed any depth-$m$ unitary circuit with arbitrary single-qubit
gates and CNOT gates into a set of paths along the monitored quantum circuit
that avoid measurements.  We then show that the accessible dimension of such unitary
circuits grows linearly in $m$, thereby showing that $d_M = \Omega(m)$.
The bulk of this section is concerned with constructing such an embedding.

When $p<p_c$, there are $\Omega(n)$ measurement-free paths even for
exponentially long monitored random quantum circuits:
\begin{lem}[Existence of measurement-free paths]
  \label{lem:linear_paths-maintext}
  If $p<p_c$, there exist $\Omega(n)$ disjoint measurement-free paths that
  percolate throughout the length of the circuit {in time $e^{\Omega(n)}$}, with probability $1-e^{-\Omega(n)}$.
\end{lem}
We give a proof of Lemma~\ref{lem:linear_paths-maintext} in
Appendix~\ref{subsec:App:percolation_largeaspect} (stated as
Lemma~\ref{lem:linear_paths}).  
%
%
%
Without loss of generality, we can assume that all measurement outcomes in the monitored circuit
are $0$ without changing the accessible dimension associated with the measurement configuration $M$.
Indeed, the gates are chosen at random from the unitarily invariant Haar measure on $SU(4)$;
thus, for any measurement outcome $1$, we can map the setting to an equivalent one
where the measurement is $0$ and where additional $X$ gates are applied immediately before
and immediately after that measurement.

We now seek to construct an embedding of a quantum unitary circuit of $\Omega(n)$ qubits into the
monitored random quantum circuit, where each measurement-free path carries one qubit of
the unitary circuit.
We first construct this embedding in a simpler situation with some additional
convenient assumptions.  We then present the embedding in the general case, lifting
all the simplifying assumptions.

Let us assume that
all measurement-free paths always traverse gates from an input leg of
the gate to an output leg of the gate.
Following such measurement-free paths, one does not go
back in the time direction, and we say the paths are \emph{causal}.
Each path is assigned to {carry one qubit while avoiding measurements}.
We apply the identity
gate or the SWAP gate so that the qubit state follows the legs of the path (Eq.~(\ref{eq:simple_example})).
In this way, qubit states are transferred along the measurement-free paths without being measured.
We can apply an arbitrary single-qubit gate to the qubit by multiplying a single-qubit gate to the
identity gate or the SWAP gate.
Let us furthermore assume
that nearest neighbour paths meet at some points, that is, nearest neighbour
paths include the legs of the same unitary gates, and the number of the unitary gates is $\Omega(t)$ for each path.
At the point two paths meet, we can apply a CNOT gate, which results in performing a CNOT gate
on the two qubits carried by the
nearest neighbour paths.
The two-qubit gates which are outside the measurement-free paths are chosen to be identity gates.
The case described above is graphically exemplified as 
\begin{equation} \label{eq:simple_example}
    \includegraphics[width=70mm]{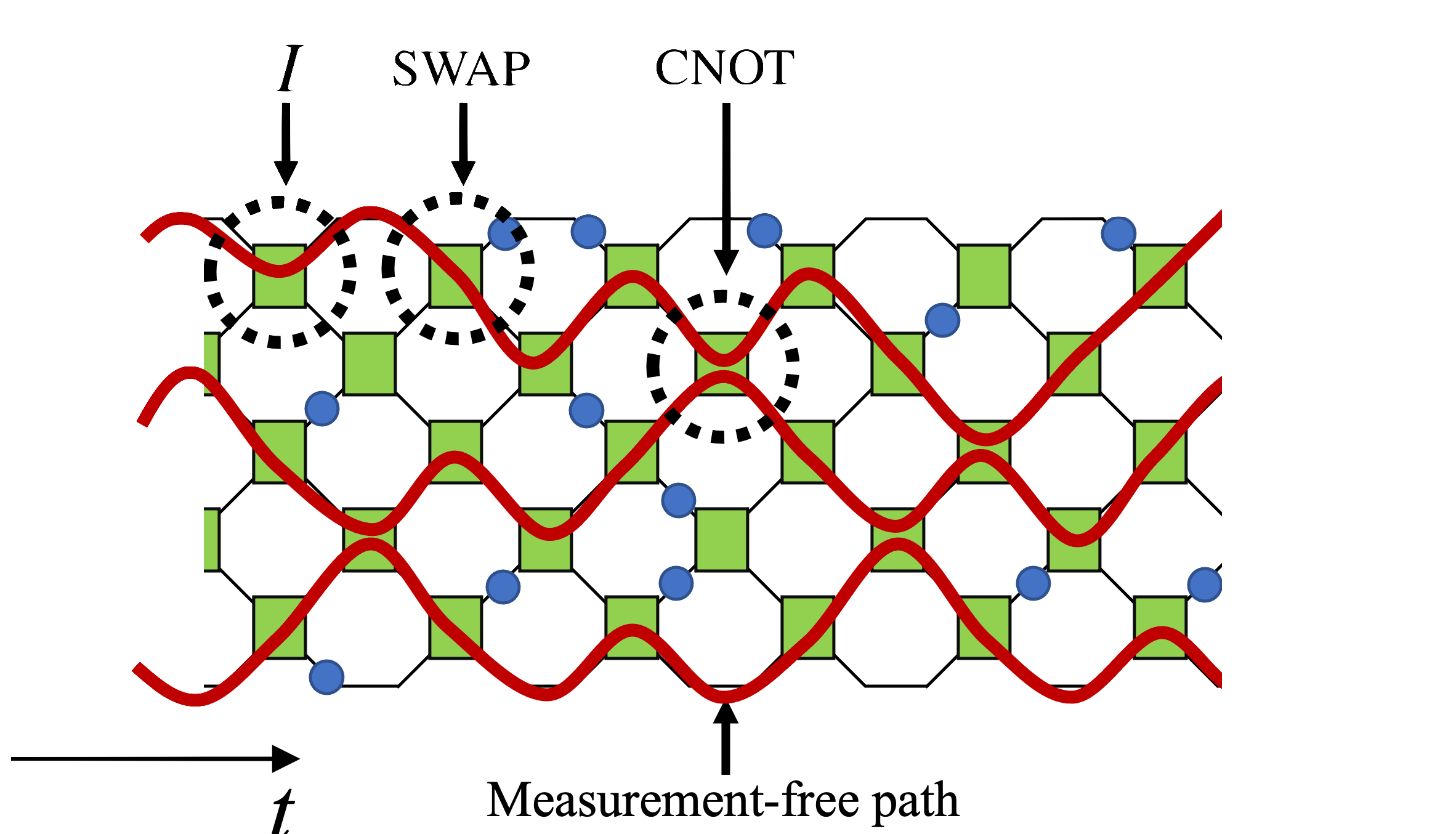}.
\end{equation}
Here, there are three such paths, i.e., it simulates a unitary circuit with three qubits, and we
apply the suitable two-qubit gate along the paths and at the points they meet, for example we
applied $I$, $\textrm{SWAP}$, $\textrm{CNOT}$ in the broken circles as shown.
Single-qubit gates can be multiplied into these two-qubit gates to enable universal
computation in the embedded unitary circuit.

In the above setting, the output states at the end of measurement-free paths is equal to
a state generated by a depth-$\Omega(t)$ unitary circuit consisting of single-qubit gates and
CNOT gates with the brick-wall architecture.
Because arbitrary single-qubit gates and CNOT gates form a universal gate set~\cite{Elementary_gate},
we can embed a universal unitary circuit into a monitored circuit with such a measurement configuration $M$. 
{Let $S^0$ be the set of the output states of a random unitary circuit with the brick-wall architecture and two-qubit random unitary gates
\begin{align} \label{eq:embedded_unitary_gate}
    U= \left( u_1 \otimes v_1 \right) W\left( u_2 \otimes v_2 \right),
\end{align}
where $u_{1,2}$, $v_{1,2}$ are Haar random single-qubit gates and $W$ is chosen from \{$I$, CNOT\} uniformly randomly.
We denote by $d_0$ the accessible dimension of $S^0$, where the accessible dimension of the random
unitary circuit is defined as per Definition~\ref{dfn:AD} with a measurement configuration that contains
no measurements and with single-qubit perturbations: $(\alpha, \beta)=(I, \sigma), (\sigma, I)$ for $\sigma \in \{X, Y, Z \}$ in Eq.~(\ref{eq:vectors}). 
The reason for the restriction of the perturbations to single-qubit gates is because only single-qubit gates in Eq.~(\ref{eq:embedded_unitary_gate}) are parametrized continuously and can be, therefore,  perturbed.
Then, because the perturbed output states of the random unitary circuit are equal to some perturbed output states of the monitored random circuit in Eq.~(\ref{eq:vectors}) 
with $M$ simulating the random unitary circuit,
up to real scalar factors, we obtain the inequality
\begin{align}
   d_M \geq d_0.
\end{align}
Then, we use an argument following Ref.~\cite{haferkamp2022linear} to lower bound $d_0$. We specify the depth of the unitary circuit by $d_0(t)$, which is $d_0$ of depth-$t$ random unitary circuits defined above. Then we prove the following lemma.

\begin{lem}[Lower bound on $d_0$]
 \label{lem:lowerbound_d0}
Let $t \geq 0$ be an integer. Then, $d_0$ grows linearly in depth as
\begin{align}
d_0(t) \geq \left\lfloor \frac{2t}{3n} \right\rfloor,
\end{align}
 until it saturates in a depth exponential in $n$.
\end{lem}
We give a proof of Lemma~\ref{lem:lowerbound_d0} in Appendix~\ref{sec:App:Accesdim} (stated as Lemma~\ref{lem:lowerbound_d02}).
Moreover, with Lemma~\ref{lem:complexity_by_dimension}, it implies that 
{$C_m(\ket{\psi})$ and $C(\ket{\psi})$}, where $\ket{\psi}$ is an output state vector of monitored random circuits with the above measurement configurations, 
also grows linearly and {at least linearly in $t$, respectively, until saturating to a value $e^{\Omega(n)}$.}

What is left to be shown is that a unitary circuit can be embedded to a monitored circuit with the general conditions, such that measurement-free paths are not always causal, and they do not meet at some space-time points.
First, we generalize the embedding to the case where measurement-free paths are not causal,
that is, the paths include the legs of two inputs or two outputs of a two-qubit unitary gate.
For now, we assume that there are measurements at the points where the path changes the time
direction (the broken circles in Eq.~(\ref{eq:measurement-free})). 
{(We discuss below how to remove this assumption using
Lemma~\ref{lem:measurement cannot increase dimension}.)}
In this case, the path is graphically shown as 
\begin{equation} \label{eq:measurement-free}
    \includegraphics[width=70mm]{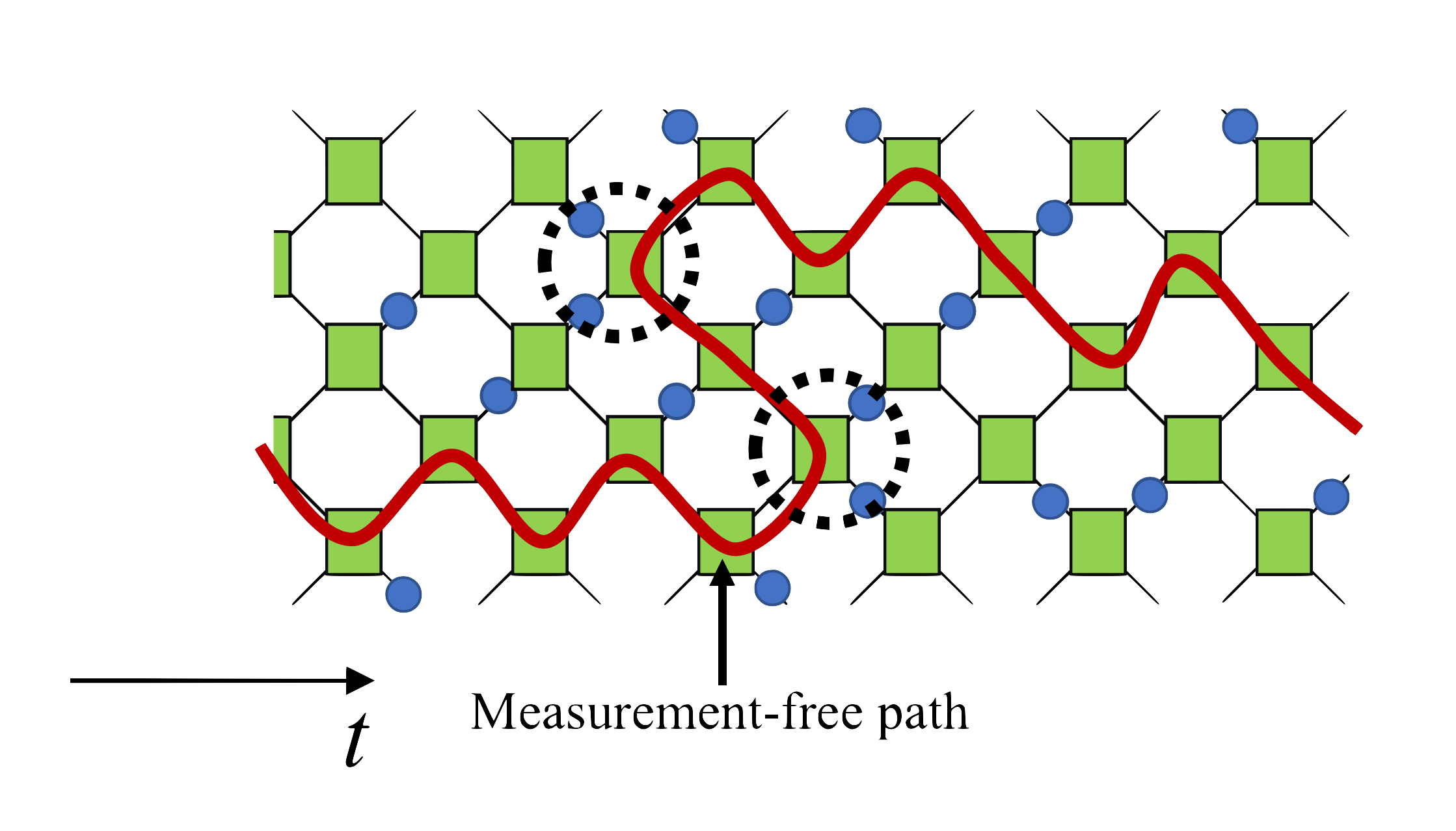},
\end{equation}
where we marked with broken circles at which the path changes the time direction. Still, the qubit state can be protected from the measurements by using a scheme similar to the entanglement teleportation. To see this, we need the following simple equality: If we choose the two-qubit gate as $U=\textrm{CNOT} (H \otimes I)$, then we have
\begin{equation} \label{eq:teleportation_3}
    \includegraphics[width=80mm]{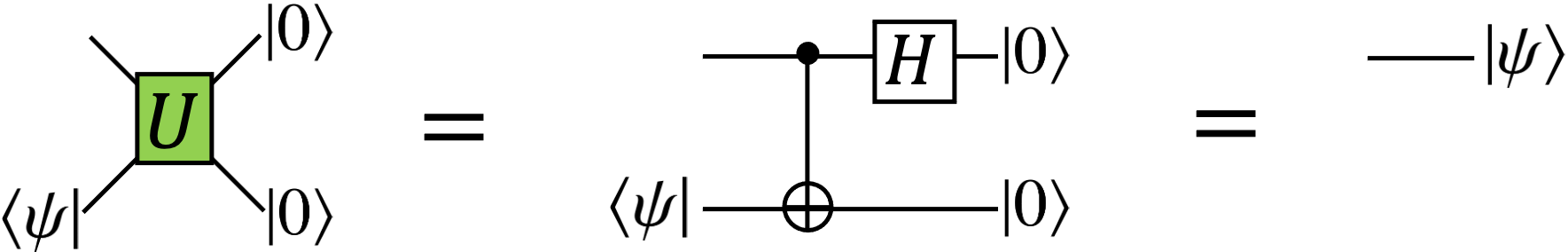},
\end{equation}
where we have omitted the constant factor in the last equality, which is not important for our proof. 
Here, we can interpret it as the measurement in the Bell basis, with a post-selection on the outcome.
With Eq.~(\ref{eq:teleportation_3}) in mind, we fix the unitary gate at which the path changes
direction {to go backwards} in time (the {bottom} broken circle in Eq.~(\ref{eq:measurement-free}))
so that the qubit state is measured in the Bell basis.  
{We also fix the unitary gate at which the path changes
direction to go forwards again in time (the top broken circle in Eq.~(\ref{eq:measurement-free})) so that
a Bell state is prepared.}
The qubit state can therefore be transferred along the path,
i.e., the input state of the measurement-free path is 
equal to its output. 

Next, we discuss how to apply a CNOT gate between two nearest-neighbour paths which do not share a unitary gate, and give a lower-bound on the number of CNOT gates that can be performed.
Here, we make use of measurement-free paths which percolate through the width of the circuit, i.e., from the top to the bottom. 
To perform a CNOT gate, we are only interested in a segment of the top-bottom measurement-free paths between the two paths carrying the quantum state, and we call such segment \emph{bridge}. They are graphically exemplified as 
\begin{equation} \label{eq:make_bridge}
    \includegraphics[width=85mm]{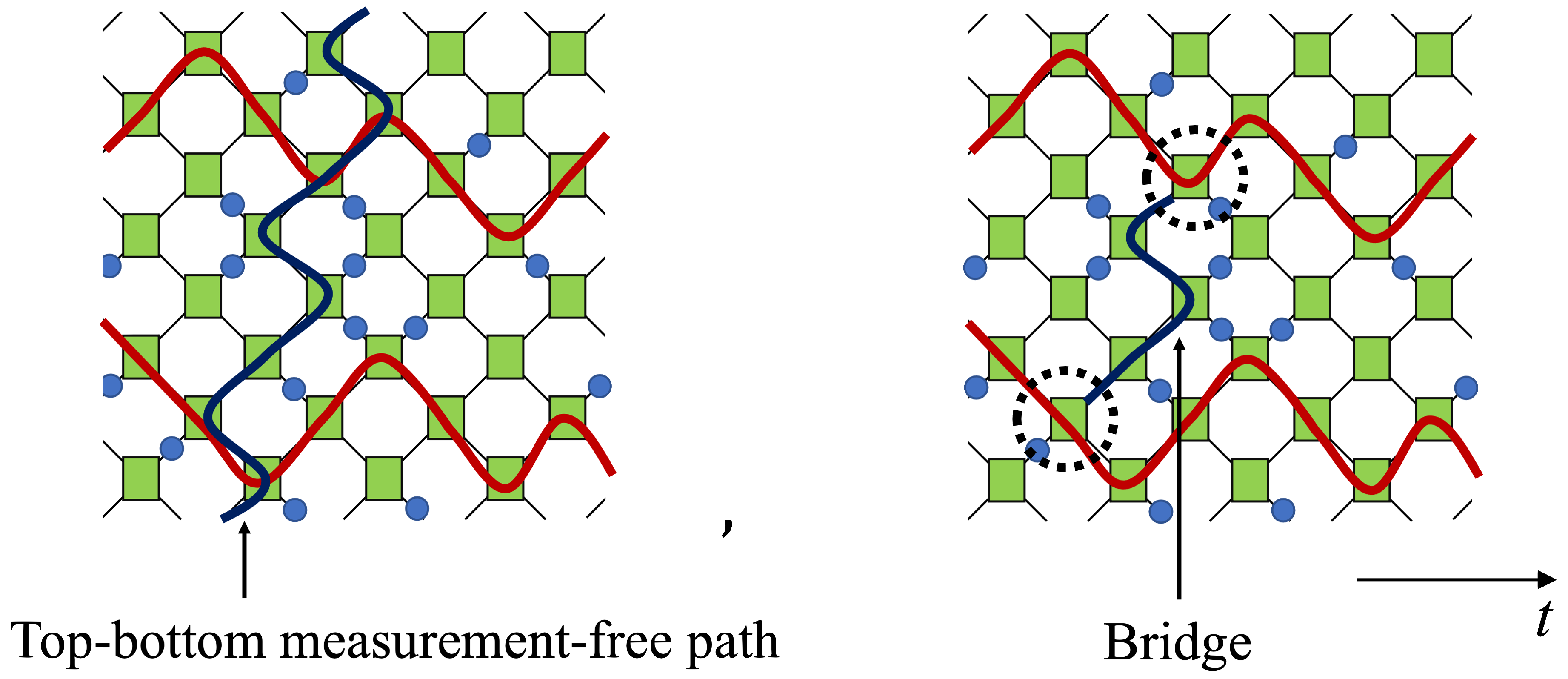},
\end{equation}
where the red lines are the paths carrying two-qubit state. 
For now, we assume again that there are measurement at the following
desired locations: (1)~the legs of the unitary gates at which the top-bottom paths
change direction in time (analogous to the broken circles in Eq.~(\ref{eq:measurement-free}))
and (2)~the fourth leg 
{of the unitary gates at the intersection of the horizontal measurement-free paths and
the top-bottom measurement-free path}
(the broken circles in Eq.~(\ref{eq:make_bridge})).
If we fine tune the unitary gates {along} 
the bridge,
we can perform a CNOT gate between nearest-neighbour measurement-free paths.
Specifically,
{we choose the unitary gates along a bridge}
such that the bridge protects a qubit state
from being measured as with the unitary gates in the measurement-free paths,
 using a SWAP gate or an identity gate if the bridge traverses the gate from an input leg
to an output leg, or the scheme in Eq.~(\ref{eq:teleportation_3}) if the bridge traverses
the gate through two input legs or through two output legs.
Also, for two unitary gates at the edge of the bridge (the broken circles in Eq.~(\ref{eq:make_bridge})),
we choose them as CNOT and $\textrm{CNOT} (I \otimes H)$ or $(I \otimes H) \textrm{CNOT}$,
{possibly multiplied by SWAP if required to ensure the qubit continues to be transferred along
the horizontal measurement-free path}. 
The target qubit of $\textrm{CNOT}$ and the order of $I \otimes H$ and $\textrm{CNOT}$  depend on the locations of legs belonging to the bridge and the path at the edge of the bridge, i.e., the shape of the path and the bridge in the broken circles in Eq.~(\ref{eq:make_bridge}).
In the example in Eq.~(\ref{eq:make_bridge}), CNOT is performed as
\begin{equation} \label{eq:Bridge_CNOT}
    \includegraphics[width=85mm]{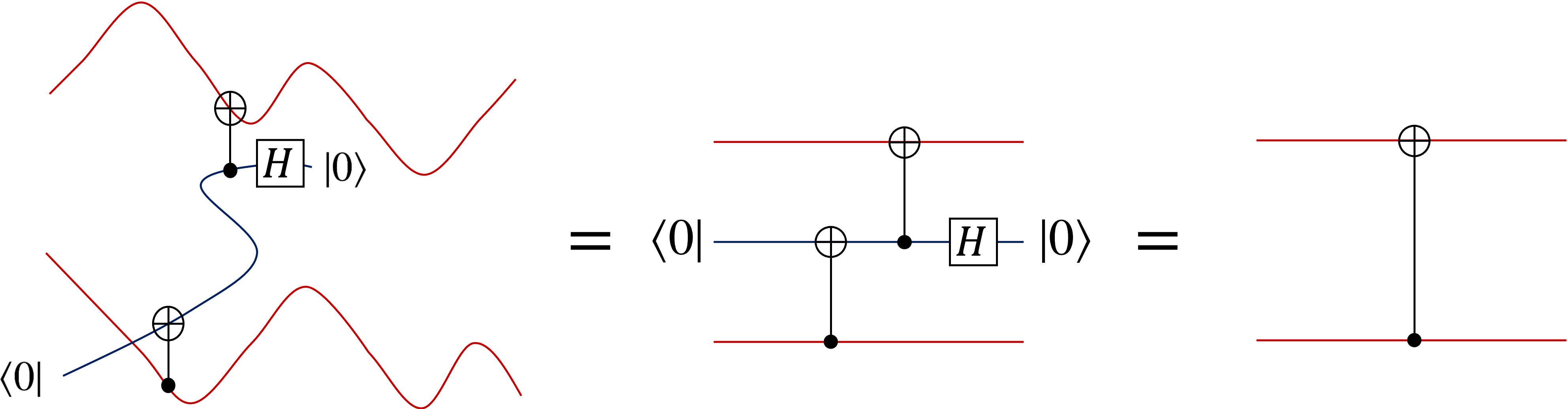},
\end{equation}
where we chose the unitary gates along measurement-free paths and inside the bridge as the specific
ones so that they carry qubit states, we apply $(I \otimes H ) \textrm{CNOT}$ and $\textrm{CNOT}$
at the edge of the bridge, and we omit the constant factor in the last equality.
 Other bridge configurations, such as if the bridge is attached on both ends to
input legs of unitary gates on the horizontal measurement-free paths, can be
treated similarly (cf. Appendix~\ref{app:two-qubit_gate}).

For every of the $n$ time steps,
that is every of the squares of the monitored circuit from $(i+1)$-th time step to $(i+n)$-th  time step, where $i$ is a multiple of $n$,
there are $O(n)$ top-bottom measurement-free paths with probability $1-e^{-\Omega(n)}$ (Fact~\ref{fct:linearnum} in Appendix~\ref{sec:App:percolation}) until an exponential number of time steps in $n$.
Then, we can apply $\Omega(t)$ layers of CNOT gates with the brick-wall architecture using the bridges made by the top-bottom paths.
We can therefore embed any depth-$m$ unitary circuit, where $m=\Omega(t)$, with arbitrary single-qubit
gates and CNOT gates into a monitored circuit with such measurement configuration.

In the discussion above, we have assumed that there are measurements at certain desired locations: around the points where the measurement-free paths and the bridges change direction in time, and on the fourth leg of each junction of the paths and the bridges.
Below, we show how a lower-bound on the accessible dimension is obtained without the measurements at the desired locations.
We consider a measurement configuration $M$ which does not include measurements at such locations.
Then, we set up another configuration $M'$ by adding measurements to $M$ at the desired locations.
Here, by adding measurements, we mean that $M'$ is made by changing some $\sqrt{1-p} I$ in $M$ to projections $\sqrt{p}\ketbra{0}{0}$ or  $\sqrt{p}\ketbra{1}{1}$.
Because we have assumed that the measurement outcomes are all $0$, we replace  $\sqrt{1-p} I$ by $\sqrt{p}\ketbra{0}{0}$.
For example, we add measurements to the points which a measurement-free path changes the time direction as
\begin{equation} \label{eq:adding_measurements}
    \includegraphics[width=85mm]{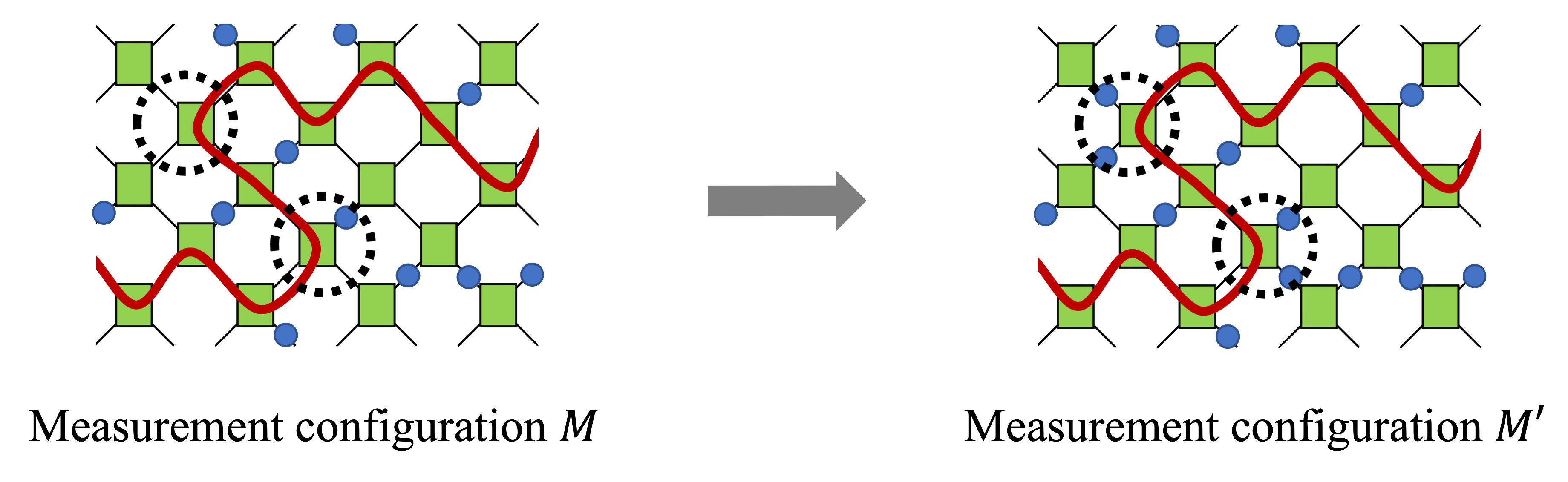}.
\end{equation}
Then, using Eq.~(\ref{eq:teleportation_3}) again, a qubit state can be transferred along the path with the measurement configuration $M'$ in Eq.~(\ref{eq:adding_measurements}).
A key lemma to lower-bound $d_M$ by considering $M'$ is that the accessible dimension
cannot increase by adding measurement (Lemma~\ref{lem:measurement cannot increase dimension}):
If $M'$ is made up by adding measurements to $M$, then $d_{M'} \leq d_M$. 
Therefore, a lower-bound on $d_{M'}$ immediately implies one on $d_{M}$.
However, adding measurements to qubits neighboring a measurement-free path might
inadvertently break another measurement-free path in the circuit.
Such a situation can occur if a measurement-free path shares a unitary gate with a
nearest-neighbour path at which it changes direction in time.
Still, the number of measurement-free paths that survive after adding the required measurements
remains $\Omega(n)$ 
because we can pick up at least half of the paths in $M$ such that any pair of two paths
do not share the same unitary gates.

In summary,
a depth-$t$ monitored circuit with $M'$, where measurement are added at the desired locations, can simulate a depth-$\Omega(t)$ unitary circuit, which implies that $d_{M'}=\Omega(t)$.
This lower-bound holds until an exponential time in $n$, because linear number of measurement-free paths in $n$ and linear number of bridges in $t$ exist until then with probability $1-e^{-\Omega(n)}$
Then, because $d_{M'}$ lower-bounds $d_{M}$, we obtain $d_{M}=\Omega(t)$, which means that the accessible dimension of a monitored circuit with measurement rate $p<p_c$ grows linearly in $t$ until a time $e^{\Omega(n)}$.

\section{Conclusion and discussion} \label{sec:conclusion}

{Our work combines techniques from quantum complexity and monitored quantum circuits to}
show that the quantum complexity of a state ---
akin to other physical quantities including the entanglement entropy
--- undergoes phase transitions in a many-body system subject to measurements. 
Our results, therefore, contribute to reinforcing the interpretation
of quantum complexity as a meaningful physical quantity,
{given its ability to identify different regimes of behavior of the evolution
of a quantum many-body system.  Indeed, the $C$- ($C_m$-) complexity undergoes a drastic
transition, depending on the rate at which measurements are applied,
between a regime where it saturates quickly and a regime in which it increases
{at least linearly until saturating to a values exponentially in the number of qubits.}
%
Our conclusions follow from rigorous mathematical arguments which do not rely on any
complexity-theoretic assumptions.


We expect our results to extend beyond the brick-wall circuit layout of Fig.~\ref{fig:intro} to more general circuit architectures. 
Given any circuit layout, the percolation properties of the corresponding graph is expected to determine the complexity phase transition of the corresponding monitored circuit.
{Our results are also anticipated to extend beyond the measurement model considered in
our work, where measurements in the computational basis occur probabilistically.}

The complexity measure $C_m$ we discuss here 
is defined with respect to
a computational model that naturally
reflects our setting, 
by 
accommodating post-selective measurements alongside unitary gates.
{A measurement outcome can be post-selected to a desired one if there is non-zero probability with which we obtain the outcome without post-selection.
Such state transformation with non-zero probability has been also discussed in the context of the state conversion by 
{stochastic local operations and
classical communication} (SLOCC) \cite{SLOCC2000, Multipartite_Entanglement}.}
Also, this computational model is 
more powerful than the computational model without post-selective measurements~\cite{aaronson2005quantum};
the measure of complexity 
$C_m$ is thus a lower bound on the usual unitary circuit complexity.  
{Our result therefore indicates that the accessible dimension is a powerful mathematical tool that can also enable us to prove linear growth of such a stronger notion of complexity, 
\je{the} $C_m$-complexity.}

Lemma~\ref{lem:measurement cannot increase dimension} provides additional insight into the
added computational power offered by measurements in monitored quantum circuits. 
It suggests that while the addition of measurements can enhance the computational
power of circuits (e.g., to prepare topologically ordered states~\cite{PhysRevLett.127.220503, tantivasadakarn2021long, PRXQuantum.3.040337, tantivasadakarn2022hierarchy} using
constant depth quantum circuits, {which is impossible without measurements})
they do not explore a set of operations that is larger when measured
in terms of accessible dimension.  As such, 
{our work offers an approach to quantify the resourcefulness
of measurements when tasked with preparing a target state on $n$ qubits.}

It is natural to consider other definitions of state complexity, such as some
approximate notion of state complexity, the {strong complexity} \cite{brandao2021models},
the {complexity entropy}~\cite{PhysRevA.106.062417}, and the
{spread complexity} \cite{PhysRevD.106.046007}.
{The strong complexity, loosely defined as the circuit size required to successfully
distinguish a state from the maximally mixed state,
displays a markedly different behavior than the $C$-complexity
in monitored random circuits.  This behavior is due to the strong complexity being
sensitive to the measurement of even a single qubit.
Indeed, for any measurement rate $p$,
the presence of a single measurement on an output
qubit resets that qubit to the state vector $\ket0$, ensuring that the output state is distinguishable
from the maximally mixed state. The strong complexity, therefore, saturates quickly for any
measurement rate in the large system size limit.}
This argument 
furthermore rules out the possibility of monitored random quantum circuits
forming a \emph{state $t$-design}~\cite{2007JMP....48e2104G}
(or complex spherical $t$-design), 
since 
forming a $t$-design implies reaching a large strong complexity~\cite{brandao2021models}.
{Moreover, our arguments agree with a recent numerical analysis indicating the absence of a}
measurement-induced phase
transition in monitored random circuits 
when judged according to the extent the
monitored random circuit approximates a 
 $t$-design; 
the latter
statement has been judged based on the
results of an application of a
machine
learning algorithm~\cite{fujii2022characterizing}. 


To make robust statements about complexity growth, one would need to smooth the
complexity measures $C(\ket\psi)$ and $C_m(\ket\psi)$ by minimizing the
corresponding complexity measure over all states that
are $\epsilon$-close to $\ket\psi$ in some reasonable metric.
Evidence points to a robust version of quantum complexity indeed growing
linearly in random circuits: arguments based on $k$-designs prove robust sublinear
growth~\cite{brandao2021models}, and variants of this method yield increasingly better properties towards robustness~\cite{haferkamp2023moments}.
Proving a similar robustness property of our results appears challenging.  It is unclear,
for instance, whether arguments based on $k$-designs can be adapted to circuits
with measurements in the general setting. 
We discuss this point furthermore in Appendix~\ref{sec:App:approximate_complexity}.
In fact, there is growing evidence that states output by a monitored
quantum circuit should have efficient representations even in some region below $p_c$ (in the
$C$-complex phase).  Indeed, numerical and analytical results~\cite{
PhysRevB.101.104301}
highlight an area law behavior 
of the R\'enyi-$\alpha$ entropies for $\alpha<1$ above
$p \approx 0.2\text{--}0.35$,
implying that such states have an efficient representation in terms of 
MPS~\cite{mpsGarcia2007,PhysRevB.73.094423}.
In this regime, a robust definition of state complexity would not exceed $\operatorname{poly}(n)$.
It remains an open problem to establish the size of the gap between robust and
exact complexity measures in this regime, as well as to determine the precise threshold
at which a robust definition of complexity grows linearly until exponential times.

Any region with $p<p_c$ where the monitored circuit's output state would nevertheless
obey an area law would provide more concrete examples of states that are naturally
described by a circuit but which have shortcuts.
Finding shorter circuits that implement a given circuit is usually hard.  The
regime $p<p_c$ is also one where we might not expect measurements to percolate across the circuit,
possibly ruling out the obvious shortcut that corresponds to the original monitored circuit
simply resetting the state to a product state at some point during its evolution.
This behavior contrasts starkly with random circuits without measurements, where such
shortcuts are not expected to occur with any significant
probability~\cite{Susskind2018PiTP_three,brandao2021models,haferkamp2023moments}.

We discuss briefly the implication of our result on
the AdS/CFT correspondence in the context of holography.
The 
``complexity=volume conjecture''~\cite{susskind2016computational} suggests
that the complexity of a CFT state corresponds to the volume of a wormhole
in the dual AdS space.
{Under the assumption that a random circuit can be regarded as a reasonable proxy to study
quantum chaotic CFT dynamics, one may}
argue that monitored random 
circuits can be seen as 
proxies of CFT dynamics with local 
measurements~\cite{PhysRevB.92.075108,numasawa2016epr,antonini2022holographic}.
Therefore, in a 
simplified model where the CFT dynamics 
is represented by a random circuit
with measurements,
our results suggest that the volume of a wormhole in the AdS space also
undergoes a phase transition by changing the holographic dual of the measurement rate.

This work invites a number of 
future research directions.
First, it would be interesting to study  the critical behaviour of the accessible
dimension in the monitored circuit in the vicinity of the critical point.
It would then be interesting to investigate if the critical exponent of the accessible dimension
agrees with that of entanglement entropy~\cite{PhysRevX.9.031009}.
Second, one could give a better lower-bound of the $C$-complexity in the complex phase.
The post-selected measurements can increase the computational power of quantum
computers \cite{aaronson2005quantum}.  
Similarly, we might expect that measurements
could increase the state complexity, which might 
grow 
faster than linearly in time.
Recently, it has been shown in
Ref.~\cite{sang2022ultrafast} that the entanglement velocity---referring to the velocity at which
a pair of well-separated regions can become entangled in 
time---%
in a monitored circuit 
below a critical measurement rate
with the maximally mixed initial state is larger than that of unitary circuits.
It would be interesting to ask if the 
state complexity grows super-linearly as well in monitored circuits at low measurement rate.
Finally, important future directions of research would address the growth of a robust
measure of quantum complexity in random monitored circuits as well as in a monitored
continuous-time evolution \cite{10.21468/SciPostPhys.7.2.024, PhysRevResearch.2.013022, PhysRevB.102.054302, PhysRevLett.127.140601, PhysRevLett.128.010603}
In particular, recent proof techniques of Ref.~\cite{haferkamp2023moments} based
on the Fourier analysis of Boolean functions appear promising to address these objectives. \je{It may also help
to use the analogy of random circuits with
the evolution under time-fluctuating Hamiltonians
\cite{RandomHamiltonians}
to establish a result of this type: After 
all, the 
latter--just like random circuits--give rise to
approximate unitary designs with high probability 
as time goes on.}
Overall, our work offers new insights on monitored quantum circuits, in which unitary dynamics and measurements
are combined together, through the lens of quantum complexity.

\section{Acknowledgements}

The authors would like to thank Keisuke Fujii and Michael Gullans for useful discussions on monitored quantum circuits, as well as Tomohiro Yamazaki and Sumeet Khatri for useful feedback on an earlier version of this manuscript. We would like to thank the DFG (CRC 183, FOR 2724), the Einstein Foundation (Einstein Research Unit on Quantum Devices), Berlin Quantum, and the FQXi for support.
We acknowledge support by the Open Access Publication Fund of Freie Universität Berlin.

\appendix

\section{Accessible dimension from algebraic geometry} \label{sec:App:Accesdim}

This section reviews the original definition of the accessible dimension based on semi-algebraic geometry and the results of Ref.~\cite{haferkamp2022linear} and discusses their extensions to monitored random quantum circuits in order to establish Lemma~\ref{lem:complexity_by_dimension}.
The facts from algebraic geometry and differential geometry and lemmas
here follow the corresponding statements in the Appendix of  Ref.~\cite{haferkamp2022linear}, where there are more detailed references.
A key observation there is that the set of the all output states $\mathcal{S}^A$ forms a semi-algebraic set, and its ``dimension'' can be
meaningfully defined and bounded, although it is neither a vector space nor a manifold.

First, we introduce some basic notions of algebraic geometry.
A subset $V \subseteq \mathbb{R}^m$ is called an \emph{algebraic set}, if for a set of polynomial maps $\{ f_j \}_{j}$, 
\begin{equation}
V= \{ x \in \mathbb{R}^m | f_j(x)=0 \ \textrm{for all} \  j \}.
\end{equation}
Also, we call a subset $W \subseteq \mathbb{R}^m$ a \emph{semi-algebraic set}, if for sets of polynomial maps $\{ f_j \}_{j}$ and $\{ g_k \}_{k}$, 
\begin{equation}
W= \{ x \in \mathbb{R}^m | f_j(x)=0, \  g_k(x) \leq 0  \  \textrm{for all} \  j \ \textrm{and} \  k\}. 
\end{equation}
The following observation is an immediate consequence of the \emph{Tarski-Seidenberg principle}, which states that for a polynomial map $F$ and a semi-algebraic set $W$, $F(W)$ is again a semi-algebraic set.

\begin{obs}[The set of output states is
semi-algebraic]
    $\mathcal{S}^M$ is a semi-algebraic set.
\end{obs}

\begin{proof}
    A set $[ \textrm{SU}(4)]^{ \times R}$  is an algebraic set, because it is the set of operators whose matrix elements satisfy polynomial equations equivalent to $U^{\dagger}U=I$ and ${\rm det}\, U=1$. 
    Besides, the contraction map $F^M$ is a polynomial map, that is, the map to output states is a polynomial function of matrix elements of $\{ U_j \}_{j=1}^R$.
    Therefore, by the Tarski-Seidenberg principle, we arrive at the stated observation.
\end{proof}

In a next step, we introduce a notion of a dimension for a semi-algebraic set. It originates from the fact that all semi-algebraic sets can be decomposed into a set 
of smooth manifolds.

\begin{fct}[Semi-algebraic sets and smooth manifolds] \label{fact:smooth_mfd}
    For a semi-algebraic set $W$, there exist a set of smooth manifolds  $\{ N_j  \}_j$  such that $W= \bigcup_j N_j$.
    Moreover, $\max_j \{ \dim (N_j) \}$ does not depend  on the decomposition of $W$.
\end{fct}

\begin{dfn}[Dimension of semi-algebraic sets] \label{dfn:dimension_semi-algebraicsets}
    For a semi-algebraic set $W$, with decomposition into smooth manifolds $W=\bigcup_j N_j$, the dimension of $W$ is defined as $d:=\max_j \{ \dim (N_j) \}$.
\end{dfn}
Using the same argument as in Lemma 1 in Ref.~\cite{haferkamp2022linear}, one can show that the above dimension of $\mathcal{S}^M$ is equal to the accessible dimension laid out in Definition \ref{dfn:AD}.

\begin{lem}[Equivalence of two definitions of dimension] \label{lem:low_rank_locus}
    Let $\dim \mathcal{S}^M$ be the dimension of the semi-algebraic set $\mathcal{S}^M$, defined by Definition \ref{dfn:dimension_semi-algebraicsets}.
    Then $\dim \mathcal{S}^M$ is equal to the accessible dimension $d_M$, defined by Definition \ref{dfn:AD}.
\end{lem}

Then, we prove Lemma \ref{lem:complexity_by_dimension}.

\begin{lem}[Restatement of Lemma \ref{lem:complexity_by_dimension}] \label{lem:complexity_by_dimension2}
    If $d_M \geq k$  for an integer $k$, then 
    \begin{equation}
    C(\ket{\psi}) \geq  \frac{1}{11} \left( k -3n -2\right), 
    \end{equation}
    $\ket{\psi} \in \mathcal{S}^A$, with unit probability, that is, for almost all unitary gates.
\end{lem}

\begin{proof}
The proof goes similarly to that of Theorem 1 in Ref.~\cite{haferkamp2022linear}, and we refer to that reference for further details. The only difference with the argument presented there is that the shorter circuit in Ref.~\cite{haferkamp2022linear} becomes a post-selected quantum circuit and the state vectors in $\mathcal{S}^M$ are not normalized in general. The latter means that unitary gates are realized with the  probability specified by the Born rule $\expval{V^M(t)^{\dagger}V^M(t)}{0^n} \prod_{i=1}^R d\mu_{\textrm{Haar}}(U_i)$, where $d\mu_{\textrm{Haar}}$ is the Haar measure on SU($4$).
The strategy is to show that for $\mathcal{S}^M$  with $d_M \geq k$, the set of states in $\mathcal{S}^M$ generated by a unitary circuit with  $R'$ two-qubit gates, which is less than $(k- n -2)/13$, is measure zero.
We explain it  in more detail below.

{
Let $\mathcal{S'}$ be the set of the all unnormalized output state vectors of a short post-selected quantum circuit consisting of $R'$ two-qubit unitary gates with an arbitrary architecture and measurement configuration.
Then, $\mathcal{S'}$ is also a semi-algebraic set.
Recall that the accessible dimension of $S'$, $d'$, is the number of linearly independent vectors of
\begin{align}
\{ F^M (U_1, \dots,  (\alpha \otimes \beta) U_j, \dots , U_R) \}_{j, \alpha, \beta}, \label{eq:vectors2}
\end{align}
where $j \in \{1, 2, \dots , R' \}$ and $\alpha$, $\beta$ $\in \{ I, X, Y, Z  \}$ are Pauli operators
such that $(\alpha, \beta) \neq (I, I)$.
The number of state vectors in Eq.~(\ref{eq:vectors2}) is at most $15R'$, and we find that $\dim (\mathcal{S'}) \leq 15R'$.
We can improve the upper-bound to $11R'+3n$, by considering the contraction of two-qubit gates $U_{i-1,i}$, acting on qubits $i-1$ and $i$, followed by $U_{i, i+1}$ that shares a bond of the gates, that is $i$-th qubit. 
Indeed, if there is no measurement on qubit $i$ just after $U_{i-1,i}$, the state vector in Eq.~(\ref{eq:vectors2}) generated by contracting the perturbed $U_{i-1,i}$, $(I \otimes \alpha) U_{i-1,i}$, and other two-qubit gates is equal to the vector generated by contracting $U_{i,i+1} (\alpha \otimes I) $ and others for any non-identity Pauli operator $\alpha$. 
It means that $3\times 2$ parameters are redundant for each two-qubit gate in a circuit's bulk. 
For the first $\frac{n}{2}$ gates, the parameters are not \je{cancelled}, 
and so are not $3n$ parameters.
Similarly if there is a projector on qubit $i$, the perturbations of $(I \otimes Z) U_{i-1,i}$ and $U_{i,i+1} (Z \otimes I)$ result in the same vector, and also the perturbations of $(I \otimes X) U_{i-1,i}$ and $(I \otimes Y) U_{i-1,i}$ result in vectors linearly dependent with each other.
It means that $2 \times 2$ parameters are redundant in this case. Therefore, we obtain
\begin{align}
    \dim (\mathcal{S'}) &\leq 15R' -3 (2R' -m)-2m +3n \\
    &=9R'+m+3n 
    \nonumber  \\
    &\leq 11R'+3n,\nonumber  
\end{align}
where $m$ is the number of measurements in the post-selected circuit, and we used $0 \leq m \leq 2R'$ in the last inequality.
}

{A quantum state vector $\ket{\psi} \in \mathcal{S}^M$ is generated by a short post-selected quantum circuit if there exists  a $\ket{\phi} \in  \mathcal{S'}$  such that $\ket{\psi}= c \ket{\phi}$ for some $c \in \mathbb{C}$.
We show  that the set of such state vectors 
\begin{equation}
\{ \ket{\psi} \in \mathcal{S}^M | \ket{\psi}= c \ket{\phi}, \  \textrm{for}  \  \ket{\phi} \in  \mathcal{S'}, c \in \mathbb{C} \} 
\end{equation}
is of measure zero in $\mathcal{S}^M$, and so its preimage by $F^A$ in $SU(4)^R$ is.
By Fact~\ref{fact:smooth_mfd}, the set of the elements of $\mathcal{S'}$ multiplied by arbitrary complex numbers, can be decomposed into smooth manifolds. Then, the maximal dimension of them is upper bounded  by $11R'+3n+2$, because complex coefficients add at most two real parameters.
Then, if $R' <  ( k -3n -2)/11$, $d_M$ is greater the dimension of the maximal manifold.
Therefore, the intersection  $\{ \ket{\psi} \in \mathcal{S}^M | \ket{\psi}= c \ket{\phi}, \  \textrm{for}  \  \ket{\phi} \in  \mathcal{S'}, c \in \mathbb{C} \}$ has Haar measure zero, because the manifolds in $\mathcal{S'}$ multiplied by arbitrary complex numbers have smaller dimensions than the maximal dimension of that in $S^M$.
This implies that the set of unitary gates in $SU(4)^{R'}$, with $R' <  ( k -3n -2)/11$, that generate states in the intersection is also Haar measure zero \cite{haferkamp2022linear}.
Because $\expval{V^M(t)^{\dagger}V^M(t)}{0^n}$ is upper-bounded by finite value, that is $1$,
it is still measure zero for the product measure of the Haar measure and the Born probability, that is $\expval{V^M(t)^{\dagger}V^M(t)}{0^n}\prod_{i=1}^R d\mu_{\textrm{Haar}}(U_i)$.
Therefore, the output states of the monitored circuit with the dimension $k$ cannot be generated by shorter quantum circuits consisting of fewer than $ ( k -3n -2)/11$ gates with unit probability, which implies the desired lower-bound of the state complexity.}
\end{proof}



Finally, we prove Lemma~\ref{lem:lowerbound_d0}.
We have considered a lower-bound for the accessible dimension of unitary circuits $d_0(t)$, with the brick-wall architecture and with depth $t$, consisting of the following random unitary gates:
\begin{align} \label{eq:embedded_unitary_gate2}
  U= \left( u_1 \otimes v_1 \right) W\left( u_2 \otimes v_2 \right),
\end{align}
where $u_{1,2}$, $v_{1,2}$ are Haar-random single-qubit gates and $W$ is chosen from \{$I$, CNOT\} uniformly randomly.

\begin{lem}[Restatement of Lemma \ref{lem:lowerbound_d0}] 
\label{lem:lowerbound_d02}
Let $t \geq 0$ be an integer. Then, $d_0$ grows linearly in depth $t$ as
\begin{align}
d_0(t) \geq \left\lfloor \frac{2t}{3n} \right\rfloor,
\end{align}
 until it saturates in a depth exponential in $n$.
\end{lem}

\begin{proof}
Recall that $d_0$ is the maximum dimension of the following vector space over unitary gates  $\{ U_1, U_2, \dots, U_R \}$ in the form of Eq.~(\ref{eq:embedded_unitary_gate2}),
\begin{align}
\{ U_R \dots (\alpha \otimes \beta) U_j \dots U_1 \ket{0^n} ) \}_{j, \alpha, \beta}, \label{eq:vectors_d0}
\end{align}
where $(\alpha \otimes \beta)$ is a single-qubit perturbation: $(\alpha, \beta)=(I, \sigma), (\sigma, I)$ for $\sigma \in \{X, Y, Z \}$.
In Ref.~\cite{haferkamp2022linear}, a unitary circuit consisting of Clifford gates is constructed in which $d_0$ grows linearly in depth. The strategy there is to construct a Clifford circuit inductively such that a linear number in $t$ of vectors
\begin{align}
     P_{\alpha, \beta, j} \ket{0^n} \in \{i^{\kappa}\ket{x} \}_{x \in \{0,1 \} ^n, \kappa \in \{0,1\} },
\end{align}
where $P_{\alpha, \beta, j} = U_1^{\dagger} \dots U_j^{\dagger} (\alpha \otimes \beta) U_j \dots U_1$, are linearly independent.
We define $C_j$ as a depth-$n/2$ Clifford circuit with arbitrary Clifford two-qubit gates.
In particular, there is a Clifford circuit such that the vectors
\begin{align}
     \{ C_1^{\dagger} \dots C_j^{\dagger} (Z \otimes I) C_j \dots C_1 \ket{0^n} \}_{j=1}^{T},
\end{align}
where $T= \lfloor \frac{2t}{n} \rfloor$, are linearly independent because of the observation that a depth-$\frac{n}{2}$ Clifford circuit is enough to turn $Z \otimes I$ into an arbitrary Pauli string by conjugating it \cite{haferkamp2022linear}.
Moreover, each two-qubit Clifford gate can be decomposed into at most three CNOT gates with single-qubit gates \cite{PhysRevA.69.010301}.
Therefore, every $\frac{3n}{2}$ time step can increase the accessible dimension at least by one, and we obtain
\begin{align}
    d_0(t) \geq \left\lfloor \frac{2t}{3n} \right\rfloor.
\end{align}
\end{proof}
The dimension $d_0$ is upper-bounded by $2 \times 2^n -1$, which is the number of real parameters in normalized quantum states, and it grows linearly until it saturates at the maximum value exponentially in $n$.

\section{Measurements cannot increase the accessible dimension} 
\label{app: measurement and dimension}


In this section, we prove that the accessible dimension of a monitored random circuit 
cannot increase by adding a projection operator. Let $M$ be a measurement configuration. 
We now construct a new measurement configuration $M'$ by changing 
an element $M_i(\tau)$  
such  that $M_i(\tau)=\sqrt{1-p}I$ into $M_i'(\tau)=\sqrt{p} \ketbra{0}{0}$ or $M_i'(\tau)=\sqrt{p} \ketbra{1}{1}$ and keeping the other elements. We denote by $|M|$  the number of 
projections in $M$, and hereafter we rename the set of projectors as $\{M_i\}_{i=1}^{|M|}$. 
We call such $M'_i$ the additional measurement. 
Then the following statement  holds.

\begin{lem}[Rank bound]\label{lem:measurement cannot increase dimension2}
For $\rank(F^{M'})$ on an arbitrary point $x' \in SU(4)^{ \times R}$, there exists a point $x \in SU(4)^{ \times R}$, on which $\rank(F^M)$ satisfies the  inequality: 
\begin{align}
    \rank(F^{M'}) \leq  \rank(F^M).
\end{align}
\end{lem}

\begin{proof}
We fix $R$ gates mapped by $F^{M'}$ as $x'=\{U_R, \dots , U_1\}$. By  definition, the rank $r'$ of $F^{M'}$ is
\begin{align}
r' = \dim \left( \textrm{span}\left\{ U_R \cdots M'_{|M'|} \cdots M'_k U_m 
\cdots (\alpha \otimes \beta) U_j \cdots M'_{1} \cdots U_1 \ket{0^n} \right\}_{j, \alpha, \beta} \right), \label{eq.states}
\end{align}
where $M'_k$ is the additional measurement $M'_k=\sqrt{p} \ketbra{0}{0}$ (we assume here that the outcome of $M_k$ is $+1$, but the case of $-1$ works as well).
 Because of $\ketbra{0}=({I+Z})/{2}$, Eq.~(\ref{eq.states}) becomes
\begin{align}
r' = \dim \left( \textrm{span}\left\{ U_R \cdots M'_{|M'|} \cdots (I + Z) U_m  \cdots (\alpha \otimes \beta) U_j \cdots M'_{1} \cdots U_1 \ket{0^n} \right\}_{j, \alpha, \beta} \right), \label{eq.states2}
\end{align}
where $U_m$ is the unitary gate which is just followed by the measurement $M_{k}$. 
By the definition of dimension, there are $r'$ linearly independent vectors $\ket{v_i} := U_R \cdots M_{|M|} \cdots
(I+Z) U_m \cdots (\alpha \otimes \beta) U_j \cdots M_{1} \cdots U_1 \ket{0^n}$, $i=1, \dots , r$, where the index $i$ denotes the configuration of $\alpha$, $\beta$, and $j$.

Now, we set $x$ as the same as $x'$ except for $m$-th gate, which is $e^{i(I+Z) \theta}U_m$.
Then, $\textrm{rank}(F^M)$ is the dimension of the vector space spanned by the vectors
\begin{align}
    & \left\{ U_R \cdots M_{|M|} \cdots
    e^{i(I+Z) \theta}U_m 
    \cdots (\alpha \otimes \beta) U_j \cdots U_1 \ket{0^n}  \right\}_{j, \alpha, \beta},
\end{align}
which are equal to 
\begin{align}    
    &\left\{ U_R \cdots M_{|M|} \cdots
    U_m \cdots (\alpha \otimes \beta) U_j \cdots M_{1} \cdots U_1 \ket{0^n} \right. \notag \\
    & \ \ +  U_R \cdots M_{|M|} \cdots  (e^{i \theta}-1)(I+Z)U_m  \left. \cdots (\alpha \otimes \beta) U_j \cdots M_{1} \cdots U_1 \ket{0^n}\right\}_{j, \alpha, \beta}. \label{eq:vecspFM}
\end{align}
Using the vectors $\{ \ket{v_i} \}$, we can find $r$ independent vectors in Eq.~(\ref{eq:vecspFM}). 
Specifically, we can find some $\theta$ such that the $r$ vectors
\begin{align} \label{eq:independent_vec_M}
    &\ket{u_i} + (e^{i \theta}-1)  \ket{v_i} \notag  \\
    & =  U_R \cdots M_{|M|} \cdots
    U_m \cdots (\alpha \otimes \beta) U_j \cdots M_{1} \cdots U_1 \ket{0^n} \notag \\
    & + (e^{i \theta}-1) U_R \cdots M_{|M|} \cdots
    (I+Z) U_m \cdots (\alpha \otimes \beta) U_j  \cdots M_{1} \cdots U_1 \ket{0^n}
\end{align}
are linearly independent for $i=1, \dots , r$, where we have defined $\ket{u_i}$ and $(e^{i \theta}-1)  \ket{v_i}$ as the first term and the second term of the right-hand side of the equation, respectively.

To see this, first, we make $r$ orthonormal vectors $\{ \ket{\Tilde{v}_i} \}$ from $\{ \ket{v_i} \}$ by the Gram-Schmidt decomposition. 
By these vectors, $\{ \ket{v_i} \}$ is decomposed as $\ket{v_i}= \sum_{k=1}^i d_i^k\ket{\Tilde{v}_k}$, for some  coefficients $d_i^k$ such that $d_i^i \neq 0$ for all $i$. Next, decompose $\ket{u_i}$ as 
\begin{equation}
\ket{u_i} = \sum_{k=1}^r c^k_i \ket{\Tilde{v}_k} + c^{\bot}_i \ket{v_i^{\bot}}, 
\end{equation}
for some coefficients $c^k_i, \  c^{\bot}_i$ and some vector $\ket{v_i^{\bot}}$ in the orthogonal complement of span$\{ \ket{v_i} \}$.
Let us define the function $f:\mathbb{R}\rightarrow \mathbb{C}$ as
\begin{equation}
f(\theta):=e^{i \theta}-1.
\end{equation}
Then, Eq.~(\ref{eq:independent_vec_M}) becomes 
\begin{align} \label{eq:independent_vec_M_decomposition}
\sum_{k \leq i} (c_i^k + f(\theta)d_i^k) \ket{\Tilde{v}_k} + \sum_{k > i} c_i^k \ket{\Tilde{v}_k} + c^{\bot} \ket{v_i^{\bot}}.
\end{align}
Again, we make $r$ orthonormal vectors $\{ \ket{\Tilde{v}^{\bot}_i} \}$ from $\{ \ket{{v^{\bot}}_i} \}$ by the Gram-Schmidt decomposition, and  $\ket{{v^{\bot}}_i}= \sum_{k=1}^i e_i^k\ket{\Tilde{v}^{\bot}_k}$ for some coeffients $e_i^k$.
Consider a linear map $A$, which maps $\ket{\Tilde{v}_i}$ to Eq.~(\ref{eq:independent_vec_M_decomposition}). 
In the matrix representation with the basis $\{ \ket{\Tilde{v}_i}, \ket{\Tilde{v}_i^{\bot}} \}_{i=1, \dots , r}$,
\begin{align}
A=
\begin{pmatrix}
   c_1^1 + f(\theta)d_1^1 & c_2^1 + f(\theta)d_2^1 & \cdots & c_r^1+f(\theta)d_r^1 \\
   c_1^2 & c_2^2 + f(\theta)d_2^2 & \cdots & c_r^2+f(\theta)d_r^2 \\
   \vdots & \vdots & \ddots & \vdots \\
    c_1^r & c_2^r & \cdots & c_r^r +f(\theta)d_r^r \\
    e_1^1 & e_2^1 & \cdots & e_r^1 \\
    0 & e_2^2 & \cdots & e_r^2 \\
    \vdots & \vdots & \ddots & \vdots \\
    0 & 0 & \cdots & e_r^r \\
\end{pmatrix},
\end{align}
where it is an $2r \times r$ matrix.
Let $A^{r \times r}$ be the top $r \times r$ sub-matrix of $A$.
Note that if the $\textrm{rank}(A^{r \times r})=r$, then $\textrm{rank}(A)=r$, and it implies that the vectors $\{ \ket{u_i} + (e^{i \theta}-1)  \ket{v_i} \}_{i=1, \dots , r} \}$ are linearly independent. 
This condition is equivalent to that $A^{r \times r}$ has a non-zero determinant.
Moreover, we can always choose $\theta$ such that $\textrm{rank}(A^{r \times r})=r$. 
This is because the determinant of $A^{r \times r}$ is a polynomial of $F(\theta)$ such that its zeros imply $\textrm{rank}(A^{r \times r})<r$, and by virtue of the fundamental theorem of algebra, the number of zeros of the polynomial is the same as its degree, which is $r$.
We can choose $\theta$ such that it is not any zeros of the  polynomials, because $\theta$ is a continuous variable.
Hence, such $\theta$ gives rank($F^M$) which is greater than or equal to rank($F^{M'}$).
\end{proof}  

Because the accessible dimension is the maximal rank over $R$ unitary gates, the above lemma implies that single-qubit measurement, or projection, cannot increase the accessible dimension.  Applying the above lemma recursively, we can show that adding any number and space-time point of measurements cannot increase the dimension.

\section{Two-qubit gate between nearest-neighbour measurement-free paths} 
\label{app:two-qubit_gate}
In this section, we show how unitary gates at the edge of the bridge are fixed to implement a CNOT gate between two neraest-neighbour measurement-free paths. For completeness, we begin with restating the method in the main text, where we consider the following paths and bridge,
\begin{equation} \label{eq:Bridge_causal_causal}
    \includegraphics[width=35mm]{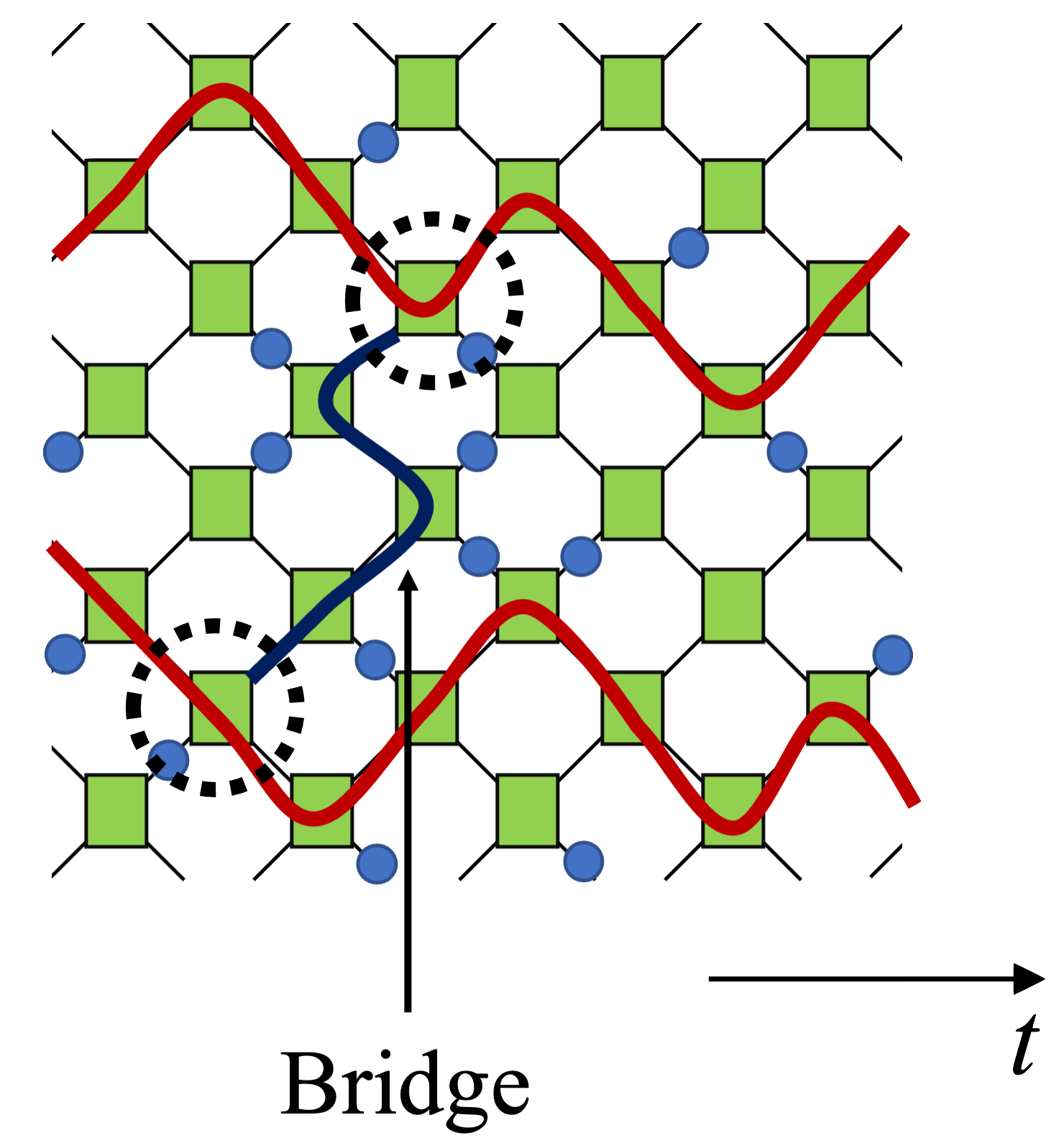}.
\end{equation}
Then, CNOT can be implemented as Eq.~(\ref{eq:Bridge_CNOT}), and we restate it here as
\begin{equation} \label{eq:Bridge_CNOT'}
    \includegraphics[width=85mm]{Bridge_CNOT.png}.
\end{equation}
In this case, both of the paths in Eq.~(\ref{eq:Bridge_causal_causal}) are causal in the broken circles, that is, they include both an input and an output of the unitary gates in the circles.
In general, in such case we can perform CNOT, by applying a CNOT gate, multiplied by $I \otimes H$, with the control qubit being measured and another CNOT gate with the target qubit state being measured at the edge of the bridge, such as Eq.~(\ref{eq:Bridge_CNOT'}).

If this is not the case, we can still implement a CNOT gate, as we explain below. We consider the case where one path is causal, and another path is not causal at the edge of a bridge, for example
\begin{equation} \label{eq:causal_acausal}
    \includegraphics[width=65mm]{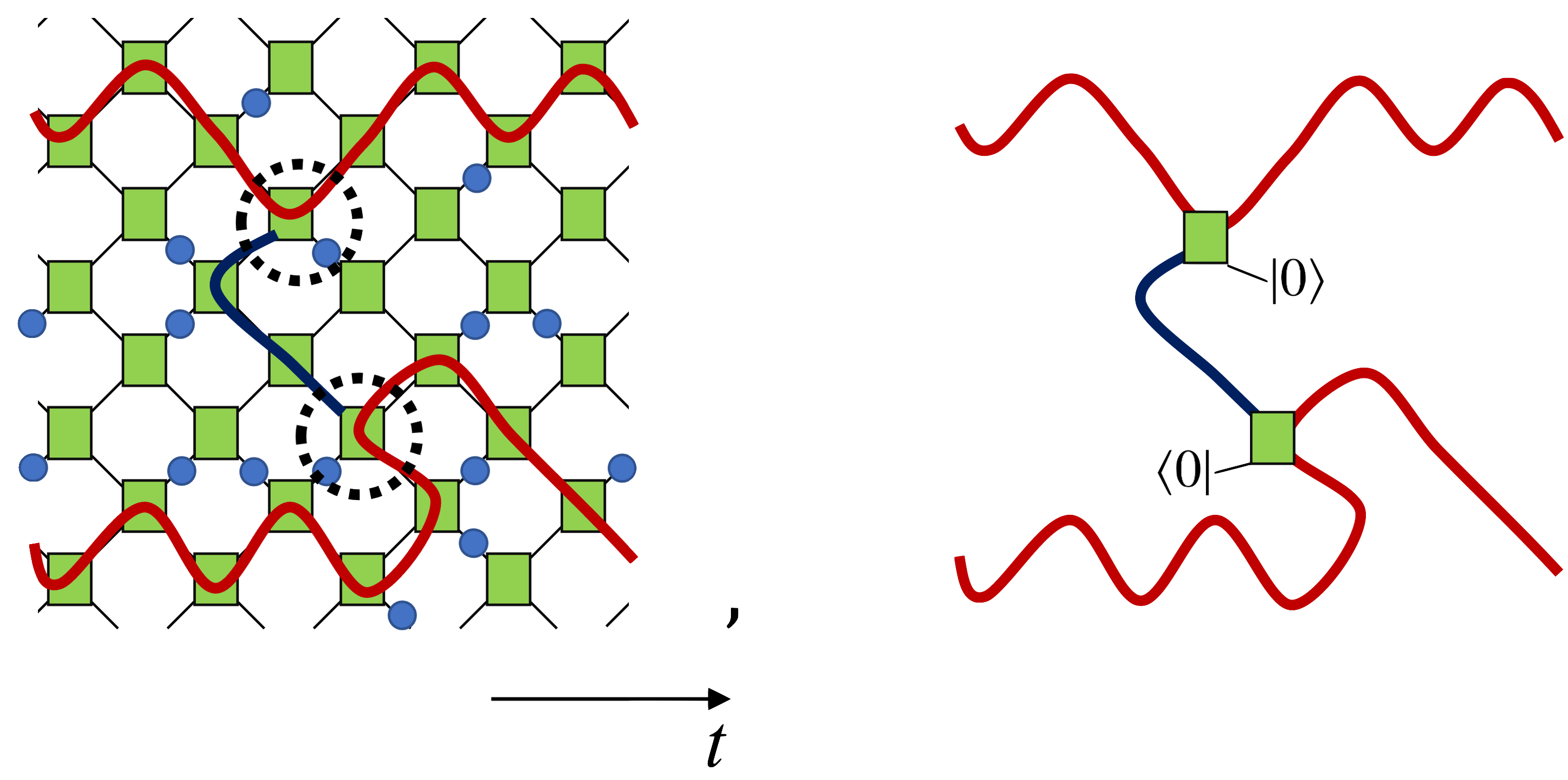},
\end{equation}
where in the right-hand side, we highlighted the paths, the bridge, and two-qubit gates at the edge of the bridge.
We can also perform CNOT in such case, by applying a CNOT gate, multiplied by $I \otimes H$, with the control qubit state being measured and another CNOT gate with the target qubit state being measured.
For the above example, it is performed as the following:
\begin{equation} \label{eq:CNOT_causal_acausal}
    \includegraphics[width=85mm]{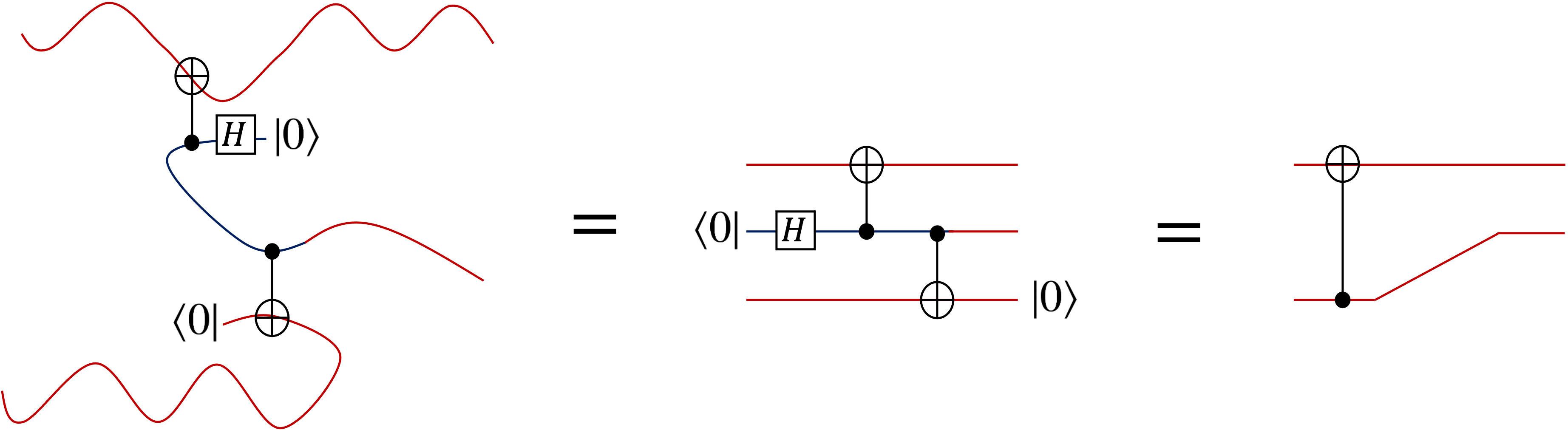}.
\end{equation}
The difference with the earlier case is that here the qubit state carried by a bridge is an output state of one measurement-free path.
Finally, we consider the 
case where both of the paths are not causal at the edge of a bridge, for example
\begin{equation} \label{eq:acausal_acausal}
    \includegraphics[width=65mm]{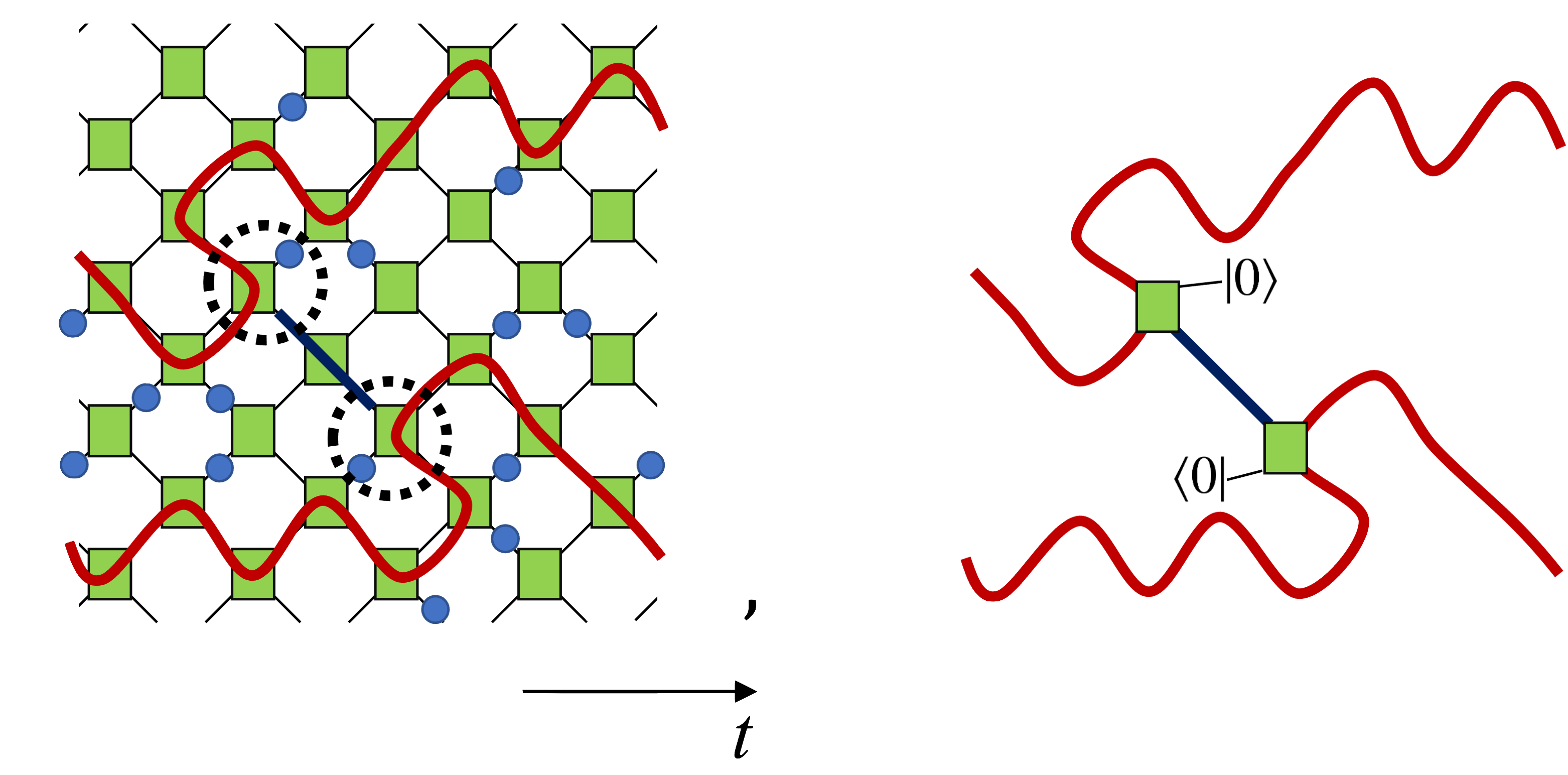}.
\end{equation}
Again, we can perform CNOT in such case, by a similar 
choice of two-qubit gates at the edge of the bridge. For the above example,
\begin{equation} \label{eq:CNOT_acausal_acausal}
    \includegraphics[width=85mm]{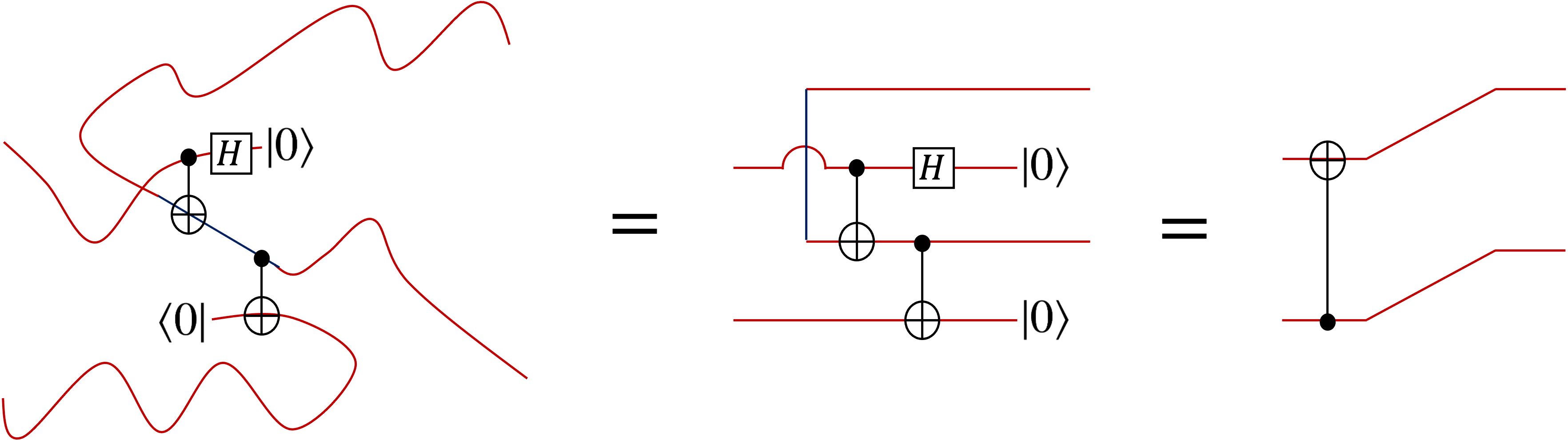}.
\end{equation}

\section{Percolation theory} \label{sec:App:percolation}


In this work, 
techniques from
percolation theory feature strongly. For this 
reason, here we
review some aspects of percolation theory, 
following Ref.~\cite{grimmett1999percolation}, and then prove lemmas used in the main text.
Specifically, we focus on the 
percolation theory on a rectangle featuring a large aspect ratio.

Especially important are notions of \emph{bond percolation} on two-dimensional square lattices.
A square lattice is defined as $\mathbb{Z}^2$ with edges between all nearest-neighbor pairs $x, y \in \mathbb{Z}^2$. We denote by $\mathbb{E}$ the set of \emph{edges}.
We define a measurable space ($\Omega$, $\mathcal{F}$) as follows. 
For the sample space, we take $\Omega = \prod_{e \in \mathbb{E}} \{ 0,1\} $, called the  edge configuration ($0$ and $1$ represent closed and open edge, respectively), and $\mathcal{F}$ is the $\sigma$-algebra on it. 
Each element in $\Omega$ is represented as a function $\omega: \mathbb{E} \rightarrow \{0,1 \}$.
We say $\omega \leq \omega'$ if $\omega(e) \leq \omega' (e)$ for all $e \in \mathbb{E}$.
Let $A \in \mathcal{F}$ be an \emph{increasing event}, i.e., 
\begin{equation} \label{Eq:increasing_events}
I_A (\omega) \leq I_A (\omega')
\end{equation}
whenever $\omega \leq \omega', \omega, \omega' \in \Omega$.
Here, $I_A:\Omega\rightarrow \{0,1\}$ 
is the \emph{indicator function} 
of $A$: $I_A (\omega)=1$ if $\omega \in A$ and otherwise $I_A (\omega)=0$.
For an event $A$, we denote the probability of the occurrence  of the event by $P_q(A)$ when an edge opens with probability $q$.
(This $q$ is contrary to that in the section 2.1. There, a measurement closes, or ``cut'' a bond, at probability ``$p$'' but in this section, bond is open at probability $q$.) 
For two increasing events 
$A$ and $B$, the 
inequality
\begin{equation}
P_q(A \cap B) \geq P_q (A) P_q (B)
\end{equation}
is well known as the 
\emph{FKG inequality} in the literature  of percolation theory.
Intuitively, the FKG inequality tells us that if we know an increasing event $A$ occurs, another increasing event $B$ is more or equally likely to occur.

Bond percolation theory is concerned with the existence or absence of \emph{left-right crossings} on a $L \times L$ square, which is an open path connecting from some vertex on the left side of the square to the right side of it. 
With probability exponentially close to one in $L$, above the critical probability $q_c$, there exists such crossings, and below it, there does not. Moreover, the critical point of bond percolation in two-dimensional square lattice is known to be $q_c=1/2$ \cite{grimmett1999percolation}.
In the following subsections, we show several lemmas to establish Theorem~\ref{thm:phase_transition_Cs}.

\subsection{Supercritical phase}
A main goal here is to derive a lower-bound of the expected number of 
{left-right crossings} on a rectangle with a various aspect ratio in the regime $q>q_c=\frac{1}{2}$.

\subsubsection{Percolation on a square} \label{subsec:App:perc.square}
We start the argument by discussing the case of a square.
Let $M_{L}$ be the maximal number of edge-disjoint left-right crossings of the box $[0,L]\times [0,L]$ for an integer $L$.
{In this appendix, we use the shorthand $[0,a]\times[0,b]$ to designate the rectangular lattice of points
of height $a$ and width $b$.}
In the supercritical phase the probability of the event $A$, where there exists an left-right crossing in the box $[0,L]\times [0,L]$, is exponentially close to one \cite{grimmett1999percolation}:
\begin{align} \label{Eq:super_critical_square}
    P_q(A) \geq 1-e^{-\alpha L},
\end{align}
for some constant $\alpha=\alpha(q)$.
The event $A$ is an increasing event, because adding open edges does not decrease the number of left-right crossings.

Now we define the \emph{interior} of $A$, $J_r (A)$, as the set of configurations in $A$ which are still in $A$ after changing arbitrarily the configurations at most $r$ edges (deleting or adding edges).
The following fact states the stability of an increasing event.

\begin{fct}
[Theorem 2.45 in Ref.~\cite{grimmett1999percolation}]
\label{fct:stability_events}
{Let $A$ be an increasing event.  Then}
 \begin{equation}
1 - P_{q_2}(J_r (A)) \leq \left( \frac{q_2}{q_2-q_1} \right) ^r(1-P_{q_1}(A))
\end{equation}
for {any} $0\leq q_1 < q_2 \leq 1$.
\end{fct}

Roughly speaking, it states that if the event $A$ happens with probability $q_1$, the modified event $J_r$ is also likely to happen when probability exceeds $q_1$. 
The above fact is useful for finding a lower-bound of the number of crossings of a rectangle.
$J_r (A)$ is the events that there exists at least $r+1$ left-right crossing (because if there are less than $r+1$ crossings, deleting $r$ edges can cut all the crossings).
Combining Eq.~\ref{Eq:super_critical_square} with Fact~\ref{fct:stability_events}, the following statement is obtained.

\begin{fct}[Lemma 11.22 in Ref.~\cite{grimmett1999percolation}] \label{fct:linearnum}
For $q > 1/2$, there exists strictly positive constants $\beta (q)$ and $\gamma (q)$, which are independent in n, such that
$P_q(M_{L} \geq \beta(q) L ) \geq 1- e^{- \gamma(q) L}$ for all $L \geq 1$.
\end{fct}

\begin{proof}
 One starts by choosing $r$ in Fact~\ref{fct:stability_events} as $\beta(q) L$, and the set $A$ as being the event that there exists at least one left-right crossing. Then Fact~\ref{fct:stability_events} implies 
 \begin{equation}
     1-e^{-L \left(\alpha(q')-\beta(q) \log\frac{q}{q-q'} \right) } \leq P_q(M_{n+1} \geq \beta(q) L ),
 \end{equation}
 where $q>q'>1/2$.
 Now we find 
 \begin{equation}
 \gamma(q)=\alpha(q')-\beta(q) \log\frac{q}{q-q'}. 
  \end{equation}
  For a fixed $1 \geq q> \frac{1}{2}$, we can choose a strictly positive constant $\beta(q)$ and $q'$ such  that  $\gamma(q)$ is also strictly positive.
\end{proof}

The above statement implies that if a left-right crossing exists at a high probability in a square lattice, we can find a number of edge-disjoint left-right crossings, which scale in the length of the side of a square, at a high probability. 
It ensures the existence of a linear number of measurement-free paths.
We mention that Refs.~\cite{browne2008phase,KielingPercolation,PhysRevLett.99.130501} have 
made a similar use of Facts~\ref{fct:stability_events} and \ref{fct:linearnum} as well in the context of the measurement-based quantum computation.

\subsubsection{Percolation on a rectangle with a various aspect ratio} \label{subsec:App:percolation_largeaspect}

Next we consider a square lattice on a rectangle $[0, L]\times [0, LT]$ for some aspect ratio $T>1$.
We can show the existence  of a scalable number of left-right crossings until some exponential aspect ratio. 
It is true in the case of both bond and site percolation.
We make use of following facts. Let $A_T$ be an event that there exists a right-left crossing on $[0, L]\times [0,LT]$ rectangle. 

\begin{fct}[Lemma 11.73 and 11.75 in Ref.~\cite{grimmett1999percolation}] \label{Fact:rectangle}
If $P_q(A_1)=\tau$, then
\begin{align}
   P_q \left( A_{\frac{3}{2}} \right) \geq (1-\sqrt{1- \tau})^3, \\
   P_q(A_{2}) \geq P_q(A_{1})  P_q \left( A_{\frac{3}{2}} \right).
\end{align}
\end{fct}

These insights 
(\emph{FKG inequality}) 
assist us in proving
 the following lemma.

\begin{lem} [Large aspect ratios]\label{lem:path_large_aspect}
{If $P_q(A_1)=\tau$, then}
\begin{align}
P_q(A_{T}) &\geq P_q(A_{1})^{T-2}  P_q \left( A_{2} \right)^{T-1} \label{eq:PATineq} \\
&\geq \tau ^{2T-3}(1-\sqrt{1- \tau})^{3(T-1)}.\nonumber
\end{align}
\end{lem}

\begin{proof}
To start with, note that if there are top-bottom crossings in every square except for both ends and  left-right crossings in every nearest-neighbor two squares, then there exists at least one left-right crossing over the entire rectangle, or, graphically,
\begin{center}
\includegraphics[width=85mm]{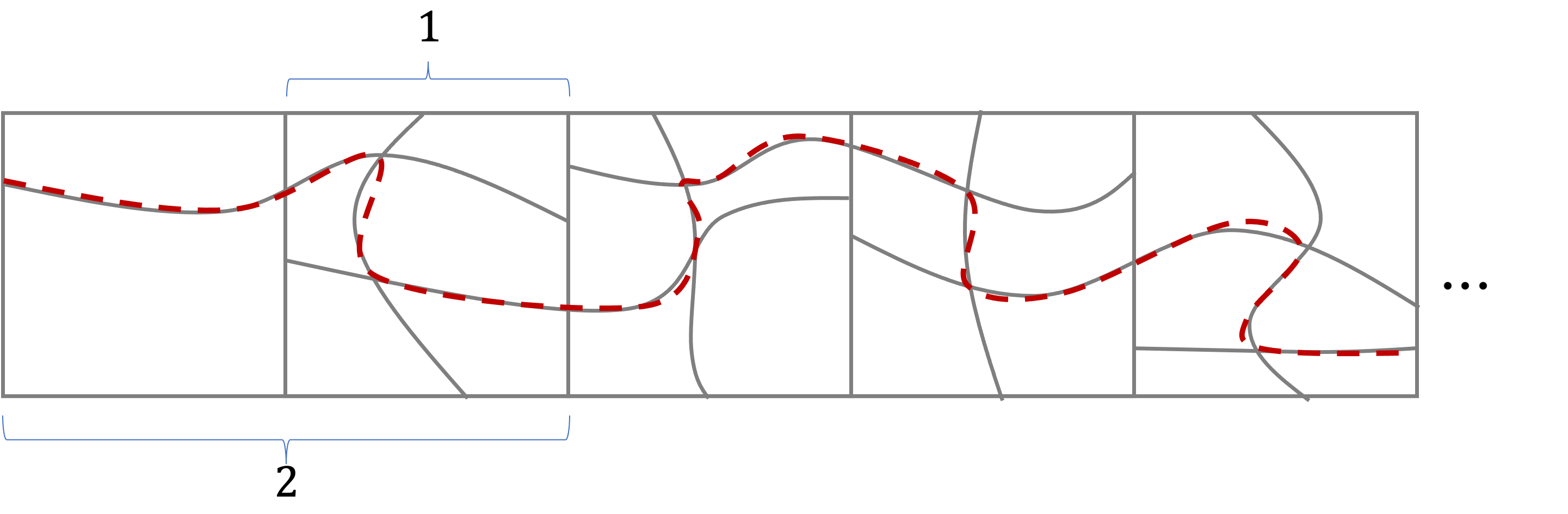},
\end{center}
where the broken line is the left-right crossing over the rectangle. 
Besides, 
\begin{align}
\left( \bigcap_{i=1}^{T-2} A_1  \right) \bigcap \left( \bigcap_{i=1}^{T-2}  A_2  \right) \subset A_T,
\end{align}
and it can be straightforwardly shown that $A_1$ and $A_2$ are increasing events.
Hence, by the FKG inequality, the inequality Eq.~(\ref{eq:PATineq}) holds, and together with Fact~\ref{Fact:rectangle}, we arrive at the validity of the second inequality as well.
\end{proof}

From the above lemma together with Fact~\ref{fct:stability_events}, we can guarantee a linear number of edge-disjoint left-right crossings.
Let $M_{L}^T$ be the maximal number of edge-disjoint open left-right crossings of the box $[0,L]\times [0,LT]$.

\begin{lem}[Linear paths] \label{lem:linear_paths}
For $q > 1/2$, there exists strictly positive constants 
$\beta (q)$ and $\gamma (q)$, which are independent {of} 
$n$, such that
\begin{equation}
P_p(M_{L}^T \geq \beta(q)L )\geq 1- e^{- \gamma(q)L} 
\end{equation}
for all $L\geq 1$.
\end{lem}

\begin{proof}
For $q> 1/2$, because of 
\begin{equation}
\tau \geq  1-e^{-\alpha(p) L}
\end{equation}
and 
 Lemma \ref{lem:path_large_aspect}, 
we obtain
\begin{align}
P_q(A_{T}) \geq & \left( 1 - e^{-\frac{\alpha(q)}{2} L  } \right)^{3(T-1)}   \times \left( 1 - e^{-\alpha(q) L  } \right)^{(2T-3)}\notag \\
\nonumber
\geq & \left( 1 - e^{-\frac{\alpha(q)}{2} L + \log3(T-1)  } \right) \notag \times \left( 1 - e^{-\alpha(q) L + \log(2T-3)  } \right)\nonumber\\
\nonumber
\geq &  1 - e^{-\frac{\alpha(q)}{2} L + \log3(T-1)}- e^{-\alpha(q) L + \log(2T-3)  }\\
\nonumber
\geq &  1 - 2e^{-\frac{\alpha(q)}{2} L + \log3(T-1)},
\nonumber
\end{align}
where, in the second inequality, we have used the Bernoulli's inequality
\begin{equation}
(1+x)^y \geq 1+xy
\end{equation}
for any real numbers $x \geq -1, y \geq 1$.
Using Fact \ref{fct:stability_events}, this implies that the number of left-right
crossings scales in the system size.
Specifically,
\begin{align}
P_q(M_{L}^T \geq \beta(q) L ) 
\geq 1 - e^{-\left( \frac{\alpha(q')}{2} -  \beta(q) \log\frac{q}{q-q'} \right)L+ \log6(T-1)}
\end{align}
where $q>q'>1/2$. 
We find that 
\begin{align}
\gamma(q)= \frac{\alpha(q')}{2} -  \beta(q) \log\frac{q}{q-q'}-\frac{1}{L}\log6(T-1).
\end{align}
We can pick a strictly positive constant $\beta(p)>0$ and $\frac{1}{2} < q' < q$ such  that  $\gamma(q)$ is also strictly positive.
\end{proof}
This lemma ensures that at below critical measurement probability, there exists a linear number of measurement-free paths until some exponential time.
Specifically, if \begin{align}
T<
\frac{e^{\frac{\gamma'(q)}{2}n}}{3} 
\end{align}
for some strictly positive 
\begin{align}
\gamma'(q)<\frac{\alpha(q')}{2} -  \beta(q) \log\frac{q}{q-q'},
\end{align}
 there exists $\lfloor \beta(q)L\rfloor$ 
 left-right crossings on a rectangle almost surely in the large $L$ limit.

\subsection{Subcritical phase} \label{App:subcritical}
Next, we consider the size of a set of connected open edges in the subcritical phase $q<q_c$.
The \emph{open cluster} $C(x)$ at a vertex $x$ of the square lattice is defined by the set of the connected open edges containing $x$, and $|C(x)|$ denotes the number of the edges in $C(x)$.
Then, the probability that the size $|C(x)|$ is large is upper bounded as follows. 

\begin{fct}[Theorem 6.75 in Ref.~\cite{grimmett1999percolation}] \label{Fact:exp_decay}
Let $C(x)$ be an open cluster containing a vertex $x$.
If $q<q_c=\frac{1}{2}$, there exists $\lambda(q)>0$ such that for any integer $k \geq 1$,
\begin{align}
P_q(|C(x)| \geq k) \leq  e^{-k\lambda(q)}.
\end{align}
\end{fct}
The independence of $x$ in the right-hand side is due to translational invariance of the square lattice.
Now, we upper-bound the probability of the event that all size of $m$ open clusters, 
\begin{align} \label{Eq:open_clusters}
 \{C(x_1), C(x_2), \je{\dots}, C(x_m) \}, 
\end{align}
is upper-bounded by an integer $k$, which we define by
\begin{align}
P_q\left( \bigcup_{i=1}^m\{|C(x_i)| \geq k\} \right). 
\end{align}

\begin{lem}[Small open clusters] \label{lem:Small_open_clusters}
For $q < 1/2$, there exists $\lambda(q)>0$ such that for any real number  $0< \epsilon <1$ and any integer $k \geq \frac{1}{\lambda(q)} \log(\frac{m}{\epsilon})$,
\begin{equation} \label{Eq:small_openclusters}
P_q\left( \bigcup_{i=1}^m\{|C(x_i)| \geq k\} \right) \leq \epsilon.
\end{equation}
\end{lem}

\begin{proof}
By the union bound and Fact~\ref{Fact:exp_decay}, we obtain
\begin{align}
P_q\left( \bigcup_{i=1}^m\{|C(x_i)| \geq k\} \right) \leq & \sum_{i=1}^m P_q\left( |C(x_i)| \geq k \right) \\
\leq & me^{-k \lambda(q)}.
\end{align}
Then, by the condition 
\begin{align}
k \geq \frac{1}{\lambda(q)}  \log(\frac{m}{\epsilon}), 
\end{align}
we obtain Eq.~(\ref{Eq:small_openclusters}).
\end{proof}

 \subsection{Tilted lattice and untilted lattice} \label{subsec:tilted and untilted}

As shown in Fig.~\ref{fig:Circuit_to_percolation} in the main text, a monitored circuit is mapped to a tilted square lattice.
In percolation theory, however, bond percolation on a square lattice is ordinarily considered for an untilted square lattice, consisting of horizontal edges. We show how they are related by proving that the critical points of them are same.

We consider bond percolation on an $2n \times 2n$ ordinary square lattice above the critical point, and the $n \times n$ tilted square whose corners are at the middle of the edges of the $2n \times 2n$ square lattice.
Then, with probability $1-e^{-\Omega(n)}$, there exist at least one left-right and one top-bottom crossings in rectangles $[0, \lfloor \frac{n}{2} \rfloor] \times [0, 2n]$  just above the middle horizontal line and right next to the middle vertical line, respectively, or graphically,
\begin{center}
\includegraphics[width=45mm]{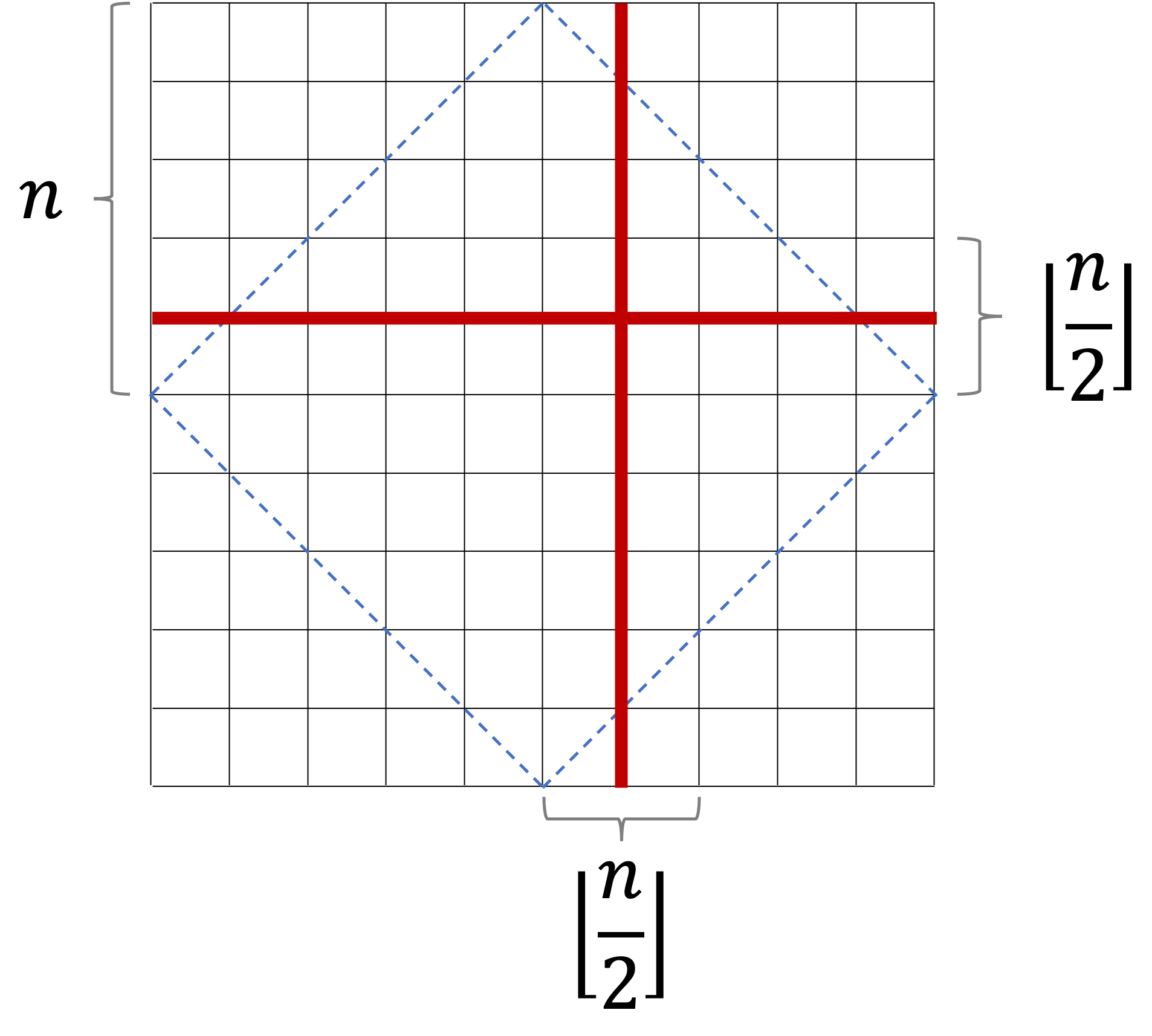},
\end{center}
where the square with broken lines is the $n \times n$ tilted square, and we assumed the crossings are straight.
It implies by the FKG inequality that there exists at least one left-right crossing in the tilted square lattice. Also, one can show conversely that left-right and top-bottom crossings in a $2n \times 2n$ tilted lattice implies a crossing in an $n \times n$ ordinary square lattice inside it by the same argument. Therefore, the critical points of percolation on both lattices are the same.

\section{Lower bound on approximate complexity} \label{sec:App:approximate_complexity}

{In this section, we give a lower bound on the approximate complexity of some class of output states in monitored random circuits at a small measurement rate.}
\begin{dfn}[$C_{\epsilon}$ approximate complexity]
    The $\epsilon$-approximate complexity of a quantum state $\ket{\psi}$, $C_{\epsilon}(\ket{\psi})$, is the minimum number of two-qubit gates to generate $\ket{\psi}$ up to an error $\epsilon$ in the trace distance from an initial product state.
    The gates can be any elements of $SU(4)$ and the circuit may have any chosen connectivity.
\end{dfn}

{To lower-bound the approximate complexity, we again use the percolation argument, which is in Section~\ref{sec:proof-Cs-complex-phase}. We consider a regime of a small enough measurement rate where there are $\Omega(n)$ causal measurement-free paths and they meet at some space-time points. We also assume that any pair of nearest-neighbour paths meet $\Omega(t)$ time so that we can embed a unitary circuit with the brick-wall architecture and circuit depth $\Omega(t)$. In this embedding, we apply the suitable two-qubit along the paths, for example we applied I, SWAP, and random two-qubit gates at the points they meet, in the broken circles as shown below:
\begin{align}
    \includegraphics[width=70mm]{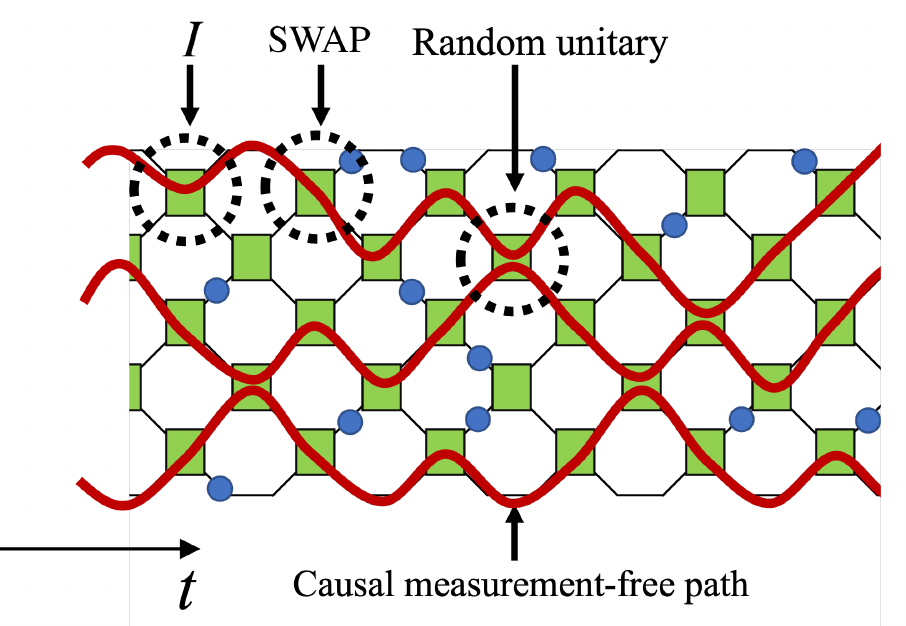}.
\end{align}
Moreover, we fix all unitary gates outside the paths as the identity gate $I$.}

Let $S'^M$ be the set of all normalized output states generated by the monitored circuit with measurement configuration $M$ and initial state $\ket{0^n}$. 
In the above setting, we can embed $\Omega(t)$-depth random unitary circuit into the monitored circuit, and we have $S'^{M}= \{ \ket{\psi} \otimes \ket{0_{\rm{rest}}} \}$, where $\ket{\psi}$ is an output state of $\Omega(n)$-qubit $\Omega(t)$-depth random unitary circuit and $\ket{0_{\rm{rest}}}$ is the state outside the measurement-free paths. 
To lower-bound on the $C_{\epsilon}$-complexity  of $\ket{\psi} \otimes \ket{0_{\rm{rest}}}$, we use the bound on the $C_{\epsilon}$-complexity of random unitary circuits as the following.
The output states of random unitary circuits with the brick-wall architecture, the circuit depth $t$, and an initial pure state form the approximate state $k$-design in $O(nk^{{5+o(1)}})$-depth with an approximation error constant in $n$ \cite{Haferkamp2022randomquantum}.
Moreover, for $\epsilon$ which is constant in $n$, the states from the approximate state $k$-design have at least $\Omega(k)$ $C_{\epsilon}$-complexity at probability $1-e^{-\Omega(n)}$ 
until depth $\Omega(2^{n})$ \cite{brandao2021models}.
Therefore, the $C_{\epsilon}$-complexity of outputs of the above monitored circuit with depth $t$ is lower-bounded by $\Omega( \left( \frac{t}{n} \right)^{\frac{1}{5+o(1)}})$ at probability $1-e^{-\Omega(n)}$ until $t\geq 2^{\Omega(n)}$.

{We note that the arguments here do not prove the $C_\epsilon$-complexity growth of a generic output state of the above monitored circuit, because we only looked at special states by choosing unitaries outside the measurement-free paths as identities. 
It becomes a nontrivial problem even if we perturb the unitaries outside the paths around identities because measurements and the following normalization can significantly enlarge the perturbation. 
We leave it an open problem to lower-bound the $C_\epsilon$-complexity in monitored circuits in the general setting.}

\printbibliography
\end{document}